\documentclass[12pt,a4paper]{article}
\usepackage{amsmath}
\usepackage{amsthm}
\usepackage{amssymb}
\usepackage{amscd}
\usepackage{graphicx}
\usepackage{epsfig}
\usepackage[matrix,arrow,curve]{xy}
\usepackage{color}
\usepackage{mathrsfs}

\usepackage{bbm}
\usepackage{stmaryrd}
\usepackage{hyperref}

\usepackage{latexsym}
\usepackage{amsfonts}
\input xy
\mathchardef\za="710B  
\mathchardef\zb="710C  
\mathchardef\zg="710D  
\mathchardef\zd="710E  
\mathchardef\zve="710F 
\mathchardef\zz="7110  
\mathchardef\zh="7111  
\mathchardef\zvy="7112 
\mathchardef\zi="7113  
\mathchardef\zk="7114  
\mathchardef\zl="7115  
\mathchardef\zm="7116  
\mathchardef\zn="7117  
\mathchardef\zx="7118  
\mathchardef\zp="7119  
\mathchardef\zr="711A  
\mathchardef\zs="711B  
\mathchardef\zt="711C  
\mathchardef\zu="711D  
\mathchardef\zvf="711E 
\mathchardef\zq="711F  
\mathchardef\zc="7120  
\mathchardef\zw="7121  
\mathchardef\ze="7122  
\mathchardef\zy="7123  
\mathchardef\zf="7124  
\mathchardef\zvr="7125 
\mathchardef\zvs="7126 
\mathchardef\zf="7127  
\mathchardef\zG="7000  
\mathchardef\zD="7001  
\mathchardef\zY="7002  
\mathchardef\zL="7003  
\mathchardef\zX="7004  
\mathchardef\zP="7005  
\mathchardef\zS="7006  
\mathchardef\zU="7007  
\mathchardef\zF="7008  
\mathchardef\zW="700A  
\newcommand{\be}{\begin{equation}}
\newcommand{\ee}{\end{equation}}

\newcommand{\lra}{\longrightarrow}
\newcommand{\bea}{\begin{eqnarray}}
\newcommand{\eea}{\end{eqnarray}}
\newcommand{\beas}{\begin{eqnarray*}}
\newcommand{\eeas}{\end{eqnarray*}}

\newcommand{\Z}{\mathbb{Z}}
\newcommand{\R}{\mathbb{R}}

\newcommand{\pa}{\partial}
\newcommand{\ti}{\times}
\def\A{{\mathcal A}}

\def\cD{{\mathcal D}}

\def\cO{{\mathcal O}}

\def\cD{{\mathcal D}}

\def\cP{{\mathcal P}}
\def\cG{{\mathcal G}}
\def\cU{{\mathcal U}}
\def\sA{{\mathsf{A}}}

\def\sD{{\mathsf D}}

\def\sJ{{\mathsf J}}

\def\sT{{\mathsf T}}

\def\sj{{\mathsf j}}

\def\sp{{\mathsf p}}
\def\s{\mathsf{s}}
\def\spi{\mathsf{\pi}}
\def\sh{\mathsf{h}}

\def\bh{{\mathbf h}}

\def\bp{{\mathbf p}}
\def\bz{{\mathbf z}}

\def\bt{{\boldsymbol \tau}}
\def\bpi{{\boldsymbol \pi}}

\newcommand{\bk}[2]{\ensuremath{\langle #1 , #2\rangle}}

\newcommand{\Bk}[2]{\ensuremath{\La #1 , #2\Ra}}

\newcommand{\La}{\big\langle}
\newcommand{\Ra}{\big\rangle}
\newcommand{\la}{\langle}
\newcommand{\ran}{\rangle}
\def\lna{\lbrack\! \lbrack}
\def\rna{\rbrack\! \rbrack}

\newcommand{\lb}{[\cdot,\cdot]}
\newcommand{\pb}{\{\cdot,\cdot\}}
\newcommand{\pair}{\bk{\cdot}{\cdot}}

\newcommand{\Ci}{C^{\infty}}

\newcommand{\mn}{{\medskip\noindent}}

\newcommand{\no}{{\noindent}}
\newcommand{\xd}{\textnormal{d}}
\newcommand{\half}{{\frac{1}{2}}}
\newcommand{\g}{\mathfrak{g}}
\newcommand{\nn}{\nonumber}
\newcommand{\ot}{\otimes}

\newcommand{\we}{\wedge}

\def\Sec{\operatorname{Sec}}
\def\Rt{{\R^\ti}}

\def\dt{\xd_{\sT}}

\def\Lst{{L^\boxtimes}}

\def\on{\operatorname}
\def\xdp{{\on{d}_\Pi}}

\def\p{\mathbf{p}}

\def\x{\mathbf{x}}
\def\h{\mathbf{h}}
\def\s{\mathfrak{s}}
\def\hs{{\mathbf{h}^*}}
\def\hh{{\widehat{\mathbf h}}}
\def\hhs{{{\widehat{\mathbf h}^*}}}
\def\hzt{{\widehat{\mathbf \bt}}}
\def\hts{{{\widehat{\mathbf \bt}^*}}}
\newcommand{\Ll}{{\pounds}}
\def\sb{\mathsf{sb}}
\newcommand{\hlra}{\lhook\joinrel\longrightarrow}

\def\z{\mathbf{z}}
\newtheorem{theorem}{Theorem}[section]
\newtheorem{proposition}[theorem]{Proposition}

\theoremstyle{definition}
\newtheorem{example}[theorem]{Example}
\newtheorem{remark}[theorem]{Remark}
\newtheorem{definition}[theorem]{Definition}

\begin{document}
\title{Jacobi algebroids\\ and Jacobi sigma models}
\author{Fabio Di Cosmo\footnote{email: fabio.di@uah.es}\\
\textit{Departamento de Física y Matemática, Universidad de Alcalá}
\\ \\
Katarzyna Grabowska\footnote{email:konieczn@fuw.edu.pl }\\
\textit{Faculty of Physics,
                University of Warsaw}
\\ \\
Janusz Grabowski\footnote{email: jagrab@impan.pl} \\
\textit{Institute of Mathematics, Polish Academy of Sciences}}
\date{}
\maketitle
\begin{abstract} The definition of an action functional for the Jacobi sigma models, known for Jacobi brackets of functions, is generalized to \emph{Jacobi bundles}, i.e., Lie brackets on sections of (possibly nontrivial) line bundles, with the particular case of contact manifolds. Different approaches are proposed, but all of them share a common feature: the presence of a \emph{homogeneity structure} appearing as a principal action of the Lie group $\Rt=\on{GL}(1;\R)$. Consequently, solutions of the equations of motions are morphisms of certain \emph{Jacobi algebroids}, i.e., principal $\Rt$-bundles equipped additionally with a compatible Lie algebroid structure. Despite the different approaches we propose, there is a one-to-one correspondence between the space of solutions of the different models. The definition can be immediately extended to \emph{almost Poisson} and \emph{almost Jacobi brackets}, i.e., to brackets that do not satisfy the Jacobi identity. Our sigma models are geometric and fully covariant.

\par


\end{abstract}
\section{Introduction}
Non-linear sigma models are field theoretic models that ubiquitously appear in physics and have been rediscovered several times over the past decades. In a nutshell, they are field theories whose fields are maps from a source manifold, possibly with a boundary, to a target manifold usually endowed with additional geometric structures, for instance, a Riemannian metric or other tensor fields. Different motivations lead to the introduction of such models. For instance, non-linear sigma models appear as effective theories of some string models, as models of lower-dimensional gravity, or as effective models to describe the dynamics of some particles.

In this paper, we are interested in non-linear sigma models defined on two-dimensional surfaces $\Sigma$ and whose target manifolds are related to Jacobi bundles. As contact structures in full generality (i.e., defined as `maximally nonintegrable' corank 1 distributions) are particular Jacobi structures, \emph{contact sigma models} form an important subclass of our Jacobi sigma models. Different action functionals have been independently proposed a few years ago \cite{Chatzistavrakidis:2020, Bascone:2021} to define Jacobi sigma models as generalizations of Poisson sigma models. Poisson sigma models were introduced at the end of the previous century by Schaller and Strobl \cite{Schaller:1994} and Ikeda \cite{Ikeda:1994} to study gravity-Yang Mills systems in reduced dimensions. Under some circumstances, this model is topological, in the sense that the space of equivalence classes of solutions under gauge transformations is a finite-dimensional manifold. A few years later, the same model was understood as a special case of a more general approach to the master equation for systems possessing BRST symmetries \cite{Alexandrov:1997}. After that, it was found out that the semiclassical expansion of the quantum correlators of this model obtained via Feynmann path integral gives Kontsevich's formula for the deformation of the algebra of functions on a Poisson manifold \cite{Cattaneo:2000}. Eventually, for suitable choices of the source surface, the space of solutions of the model was shown to be the symplectic groupoid of the Poisson manifold \cite{Cattaneo:2001}. From a physical perspective, the Poisson sigma model has found several applications also to the description of string background and branes (see for instance \cite{Calvo:2005, Calvo:2006} ).

Therefore, an obvious question arises, whether a Jacobi sigma model can be defined and in what sense it generalizes the properties of the Poisson sigma model. The approaches previously mentioned have differences: in the paper \cite{Chatzistavrakidis:2020} a la `Kaluza-Klein' approach is adopted to define the action functional of the model, whereas in the work \cite{Bascone:2021} a `constrained' approach is considered. In the first case, the action functional is interpreted, in the end, as an action functional of the Poisson sigma model for the Poissonization of the Jacobi structure, whereas, in the second situation one of the Euler-Lagrange equations is a constraint for the fields under analysis. Another difference appears when one looks at the equivalence classes of solutions under gauge transformations, since they have different dimensions. On the other hand, there are also some similarities, since in both cases the authors were interested in Jacobi brackets for functions on the target manifold of the sigma model and in this sense, the Jacobi sigma model was considered as a generalization of the Poisson sigma model.

\medskip In this paper, however, we will consider a more general situation, that of a Jacobi bundle instead of a Jacobi structure $(\zL,E)$. We will come back to this point later in the paper, but it is worth anticipating here that in the sense of Kirillov \cite{Kirillov:1976}, a \emph{Jacobi bracket} is a bracket on sections of a line bundle $L\to M$, and it is expressed in terms of a skew-symmetric bi-differential operator. Such a general description is also called a \emph{Jacobi bundle} \cite{Marle:1991}. When the line bundle is trivial, its sections can be identified with functions on the base space and one recovers the description of a Jacobi bracket in terms of a bivector field $\Lambda$ and a vector field $E$ on a manifold $M$. 
The corresponding Jacobi bracket of functions on $M$ reads
\[\{f,g\}=\zL(\xd f,\xd g)+fE(g)-gE(f).
\]
Such structures have been called by Lichnerowicz \cite{Lichnerowicz:1978} \emph{Jacobi manifolds} and used in the previous description of the Jacobi sigma model. In order to have a Jacobi sigma model for a (generally nontrivial) Jacobi bundle, we follow here the approach to Jacobi brackets elaborated by some of the authors (see, for instance, \cite{Bruce:2017}). This approach fits much better to the module, rather than the algebra character of Jacobi brackets from a purely algebraic point of view. For trivial line bundles, i.e., Jacobi brackets on $\Ci(M)$, it is easy to mix the $\Ci(M)$-module structure on $\Ci(M)$ with the associative algebra structure on $\Ci(M)$.

According to this point of view, the structure of a Jacobi bundle $(L,\pb)$ can be equivalently described as a principal $\Rt$-bundle $P\to M$ endowed with a Poisson tensor $\zP$ which is homogeneous of degree $-1$, where the homogeneous structure is understood in terms of the action of the structure group $\Rt$. If the Poisson tensor $\zP$ is invertible, then $\zw=\zP^{-1}$ is a 1-homogeneous symplectic structure, and we deal with a contact structure on $M$. In fact, a canonical choice of $P$ is $P=\big(L^*\big)^\ti$, where, for a vector bundle $V$, with $V^\ti$ we denote $V$ with the zero-section removed, i.e., the open submanifold of non-zero vectors in $V$. Let us recall that
by $\Rt$ we mean the multiplicative group $\R \setminus \{0\}=\on{GL}(1;\R)$ of non-zero reals. The $\Rt$-bundle
$\big(L^*\big)^\ti$ will be denoted shortly $\Lst$.

Once recognized the r\^ole of the homogeneous Poisson structure, one defines the Jacobi sigma model for a Jacobi bundle as a Poisson sigma model with the target manifold being the principal $\Rt$-bundle $\Lst$, and a homogeneous Poisson bivector field $\zP$. The Hamiltonian description of the system allows us to identify the space of equivalence classes of solutions under gauge transformations with a symplectic $\Rt$-groupoid endowed with a homogeneous multiplicative symplectic structure of degree 1. Therefore, as one would have expected, the space of solutions of the Jacobi sigma model for certain source manifolds and proper boundary conditions possesses the structure of a \emph{contact groupoid} in the sense of \cite{Bruce:2017}.
Adopting a more geometric perspective, the solutions of the Poisson sigma models are Lie algebroid morphisms $$\zF:\sT\Sigma\to\sT^*\Lst,$$
where $\sT^*\Lst$ is the Lie algebroid associated with the Poisson manifold $\Lst$ (see for instance \cite{Bojowald:2005}).

In this case, however, the cotangent bundle $\sT^*\Lst$ is not only a Lie algebroid but carries also an $\Rt$-bundle structure, which turns $\sT^*\Lst$ into a \emph{Jacobi algebroid}, i.e., the group $\Rt$ acts by automorphisms of the Lie algebroid structure. It seems therefore reasonable to include the $\Rt$-action into the picture and consider fields to be a morphism of $\Rt$-vector bundles
$$\Psi:\sT\Sigma\ti \Rt\to\sT^*\Lst\,,$$
expecting that solutions will be morphisms of Jacobi algebroids.
However, an unexpected discovery (Theorem \ref{the}) is that $\Psi$ is a morphism of $\Rt$-Lie algebroids if and only if its restriction to $\sT\zS$, say $\zF$, is a Lie algebroid morphism, so the solutions for both models are essentially the same.

On the other hand, one can take the above viewpoint and pass to the reduced data, namely to fields consisting of
vector bundle morphisms
$$\Psi_0:\sT\zS\to\sJ^1L=\sT^*\Lst/\Rt,$$
where $\sJ^1L$ is the first jet bundle of sections of $L$, and maps $\psi:\zS\to\Lst$. The latter map is equivalent to a regular morphism of line bundles $\zf:\zS\ti\R\to L$, i.e., a VB-morphism which is regular on fibers (here: isomorphism on fibers).

Undoubtedly, the proposed models share some features with the `a la Kaluza-Klein' approach. Indeed, when considering the case of a Jacobi structure on functions of the target manifold, the action functional in \cite{Chatzistavrakidis:2020} almost coincides with the one proposed in this paper. However, the construction of the model is very different, since the action functional here is postulated on the basis of a geometric analysis of Jacobi bundles, whereas in \cite{Chatzistavrakidis:2020} the action functional is derived from gauge symmetry reasons. One could also interpret the proposal for the Jacobi sigma model advanced in this paper as another example of unfolding and reduction of physical systems (see \cite{Carinena:2015} for more details about the subject): even if locally one could describe the bracket in terms of a bivector field $\Lambda$ and a vector field $E$ on $M$, the proper geometrical structure characterizing a Jacobi bundle is a tensor field on an enlarged space with certain properties under a suitable group action. Then, the quotient under the $\Rt$-action allows to obtain the `reduced' model in terms of fields taking values in the first jet-bundle $\sJ^1 L $ of the vector bundle $L$ (the construction will be outlined in Section \ref{HomPoi}).

Let us consider now, with more details, the structure of the paper. In Section 2, we fix the notation, and Section 3 contains basics on Lie algebroids and Lie algebroid morphisms. In Section 4, the properties of $G$-algebroids and their morphisms are presented in their general form, as a preparation for the description of Jacobi algebroids. In sections 5 and 6, the relation between $\Rt$-principal bundles and line bundles is presented, and the concept of the \emph{category of line bundles} is defined. In this framework, some lifting procedures are described which will be used in the rest of the paper. In sections 7 and 8, we introduce the concepts of Jacobi bundles and Jacobi algebroids; these will be the basic ingredients to construct an action functional for the Jacobi sigma model. Section 9 contains the main results of the paper: three proposals for an action functional describing the Jacobi sigma model are advanced and it is shown that the corresponding solutions are in one-to-one correspondence with each other. In section 10 we extend the previous model to the case of almost Poisson and almost Jacobi structures, where the brackets among functions and sections are not required to satisfy the Jacobi identity. However, while this model is constructed via a different bivector field, in literature almost Poisson and almost Jacobi structures appear from the coupling with WZW models including a non-closed H-field. Some concluding remarks complete the work.

\section{Notation}
All objects we will consider are in smooth categories. Let us fix some notation.

\mn If $M$ is a manifold and $\sp:E\to M$ is a fiber bundle over $M$, then
\begin{itemize}
\item with $\Ci(M)$ we denote the algebra of (real) smooth functions on $M$;
\item with $\zt_M:\sT M\to M$ we denote the tangent bundle of $M$;
\item with $\zp_M:\sT^*M\to M$ we denote the cotangent bundle of $M$;
\item with $\Sec(E)$ we denote the space of sections of $E$. If $E$ is a vector bundle, then $\Sec(E)$ is canonically a $\Ci(M)$-module.  With $\mathcal A^i(E)=\on{Sec}(\wedge^iE)$ we will denote the $\Ci(M)$-module of sections of the vector bundle bundle $\wedge^iE$, and with
$$\mathcal A(E)=\bigoplus_{i\in \mathbb Z}\mathcal A^i(E)$$
the Grassmann algebra of multisections of $E$. Here, of course, $\A^0(E)=\Ci(M)$ and $\A^i(E)=\{0\}$ for $i<0$.
\item if $E$ is a vector bundle, then with $\sp_*:E^*\to M$ we denote the dual vector bundle, and for a section $\ze$ of $E$ with $\zi_\ze$ we denote the corresponding linear function on $E^*$;
\item with $\sj^1\sp:\sJ^1E\to M$ we denote the bundle of first jets of sections of $E$. If $E$ is a vector bundle, then $\sJ^1E\to M$ is a vector bundle. In this case, sections of the dual bundle $\sj^1_*\sp:\big(\sJ^1E\big)^*\to M$ represent linear first-order differential operators  $\sD:\Sec(E)\to\Ci(M)$.
    If $F\to M$ is another vector bundle over $M$, then sections of $\big(\sJ^1E\big)^*\ot_MF$ represent linear first-order differential operators $\sD:\Sec(E)\to\Sec(F)$;
\item if $\ze$ is a section of $E$, then with $\sj^1\ze\in\Sec(\sJ^1E)$ we denote the first-jet prolongation of $\ze$.
\end{itemize}

\section{Lie algebroid morphisms}
Let $\zt:E\to M$ be a vector bundle and $\zp:E^*\to M$ be its dual. If $(x^a)$ are local coordinates in an open subset $U\subset M$, then we will use affine coordinates $(x^a,\zx_i)$ on  $\zp^{-1}(U)\subset E^*$, and the dual coordinates $(x^a,y^i)$ on $\zt^{-1}(U)\subset E$, as associated with dual bases, $(e_i)$ and $(e^i)$, of local sections of $E$ and $E^*$, respectively.

\begin{definition}  A {\it Lie algebroid} structure on $E$ is given by a linear Poisson tensor
$\zP$ on $E^*$.
\end{definition}
\no Let us recall that the linearity of $\zP$ means that it is homogeneous of degree $-1$ with respect to the Euler vector field $\nabla=\zx_i\pa_{\zx_i}$ on $E^*$ or, equivalently, that the corresponding Poisson bracket $$\{f,g\}_\zP=\zP(\xd f,\xd g)$$
is closed on linear functions. In local coordinates,
\be\label{SA} \Pi =c^k_{ij}(x)\zx_k
\partial _{\zx_i}\otimes \partial _{\zx_j} + \zr^a_i(x) \partial _{\zx_i}
\wedge \partial _{x^a}\,,
\ee
where $c^k_{ij}=-c^k_{ji}$. We will call the local functions $c_{ij}^k$ and $\zr^a_i$ the \emph{structure coefficients} of $\zP$.
\begin{theorem}\label{Lal}
A Lie algebroid structure $(E,\Pi)$ can be equivalently defined
\begin{itemize}
\item as a Lie bracket
$[\cdot ,\cdot]_\Pi $ on the space $\Sec(E)$ of sections of $E$, together with a vector bundle morphisms\ $\zr
\colon E\rightarrow T M$ covering the identity on $M$ (the \emph{anchor}), such that
\be\label{qd} [X,fY]_\Pi =\zr(X)(f)Y +f [X,Y]_\Pi,
\ee
for all $f \in C^\infty (M)$, $X,Y\in \Sec(E)$. The Lie algebroid bracket $[\cdot,\cdot]_\Pi$ and the anchor $\zr$ are related to the Poisson bracket $\{\cdot,\cdot\}_{\zP}$
according to the formulae
\beas
        \zi_{[X,Y]_\Pi}&= \{\zi_X, \zi_Y\}_{\zP},  \\
        \zp^*\big(\zr(X)(f)\big)       &= \{\zi_X, \zp^*f\}_{\zP}\,.
                                                   \eeas
\item or as a homological derivation $\xdp:\A(E^*)\to \A(E^*)$ of degree 1 in the Grassmann algebra
$\A(E^*)$ (the {\it Lie algebroid de Rham differential}), i.e.,
$\xdp$ maps $\A^i(E^*)$ into $\A^{i+1}(E^*)$, satisfies the graded Leibniz rule
\[\xdp(\za\we\zb)=\xdp\za\we\zb+(-1)^a\za\we\xdp\zb\,,\]
for $\za\in\A^a(E^*)$, and satisfies the homological condition $\xd^2_\zP=0$. Such a Lie algebroid de Rham differential is uniquely determined by its action on functions on $M$ and sections of $E^*$ (linear functions on $E$) given by
$$\bk{\xdp f}{X}=\zr(X)(f),\quad \xdp\za(X,Y)=\zr(X)(\Bk{\za}{Y})-\zr(Y)(\Bk{\za}{X})-\Bk{\za}{[X,Y]}.$$
\end{itemize}
\end{theorem}
In a local basis of sections and the corresponding local coordinates, we have
$$[e_i,e_j]_\Pi=c^k_{ij}e_k,\quad \zr(e_i)=\zr^a_i\pa_{x^a}$$
and
\[\xdp f=\zr^a_i\frac{\pa f}{\pa x^a}e^i,\quad \xdp e^k=\half c^k_{ji}\,e^i\we e^j\,.
\]
Identifying local sections $e^i$ of $E^*$ with local linear functions $y^i=\zi_{e^i}$ on $E$, we will also write
\be\label{drd}\xdp f=\zr^a_i\frac{\pa f}{\pa x^a}y^i,\quad \xdp y^k=\half c^k_{ji}\,y^i\we y^j\,.
\ee
An explicit form of $\xdp$ is given by the following formula, analogous to the well-known Cartan's formula for the standard de Rham derivative:
\beas && \xd_\zP\zm(X_1,\dots,X_{k+1}) = \sum_i (-1)^{i+1}\zr(X_i)\big(\zm(X_1,\dots,\widehat{X}_i,\dots, X_{k+1})\big)\\
&&+\sum_{i<j}(-1)^{i+j}\zm\big([X_i,X_j]_\Pi,X_1,\dots,\widehat{X}_i,\dots,\widehat{X}_j,\dots,X_{k+1}\big)\,.
\eeas
\begin{remark}
$\R$-linear operators $D:\Sec(E)\to\Sec(E)$, acting on sections of a vector bundle $E\to M$ and satisfying
the condition
\be\label{Qder} D(fY)=fD(Y)+\hat D(f)Y,
\ee
where $\hat D$ is a (uniquely determined) vector field on $M$, are known in the literature under various names, e.g., \emph{$Q$-derivations, derivative endomorphisms, covariant differential operators} (see \cite{Grabowski:2003} for a short survey and references). We will use the name \emph{vector bundle derivations} (\emph{VB-derivations} in short). Of course, VB-derivations are particular linear first-order differential operators. They can be identified with linear vector fields on the dual bundle $E^*$. This identification corresponds to the identification of sections $\zs$ of the vector bundle $\zt:E\to M$ with linear functions $\zi_\zs$ on $E^*$. VB-derivations are sections of a vector bundle $\sD E$ over $M$ called the \emph{gauge algebroid of $E$}, since it is a canonical Lie algebroid associated with the vector bundle. Of course, VB-derivations are particular linear first-order differential operators from $E$ to $E$, so
$$\sD E\subset\big(\sJ^1E\big)^*\ot_ME.$$
If $E$ is of rank 1 (a line bundle), then the above inclusion is the identity.

Note that (\ref{Qder}) can be rewritten in the form
$$[D,m_f]=D\circ m_f-m_f\circ D=m_{\hat D(f)},$$
where $m_f$ is the operator in $\Sec(E)$ of multiplication by $f$.
In fact, to characterize VB-derivations, it is enough to require that $[D,m_f]=m_g$ for some $g\in\Ci(M)$. The fact that $g=X(f)$ for some vector field $X$ comes automatically (see \cite[Theorem 1]{Grabowski:2003}).

\mn The property (\ref{qd}) of Lie algebroid brackets means exactly that the adjoint operators $\textrm{ad}_X=[X,\cdot]_\zP$ are VB-derivations. In fact, if the vector bundle $E$ has rank $>1$, then a Lie algebroid structure on $E$ can be simply defined as a Lie bracket $\lb$ on the space $\Sec(E)$ such that $\textrm{ad}_X$ is a VB-derivation for each $X\in\Sec(E)$ \cite{Grabowski:2003}.

\no Note also that the de Rham derivative $\xdp$ defines a Lie algebroid cohomology in the standard way. In the case of a Lie algebra, $\xdp$ is the \emph{Chevalley-Eilenber cohomology operator}, whereas $\xdp$ is the standard de Rham differential in the case of the canonical Lie algebroid $E=\sT M$.
\end{remark}

\mn To define what is a morphism of Lie algebroids one can try to follow the Lie algebra case and say that it is a morphism of the corresponding vector bundles
\be\label{Lam}
\xymatrix@C+10pt{
E_1 \ar[r]^{\zF}\ar[d]_{\zt_1} & E_2\ar[d]^{\zt_2} \\
M_1 \ar[r]^{\zf} & M_2}
\ee
respecting the Lie algebroid bracket. However, this causes problems, as the morphism (\ref{Lam}) does not generally define any push-forward of sections. On the other hand, (\ref{Lam}) defines properly a pull-back of sections of the dual bundles, so a map $\zF^*:\A(E^*_2)\to\A(E^*_1)$. The corresponding de Rham differentials completely determine the Lie algebroid structures, so an obvious solution is the following.
\begin{definition}
A \emph{morphism of Lie algebroids} $(E_1,\zP_1)\to(E_2,\zP_2)$ is a morphism (\ref{Lam}) of the corresponding vector bundles such that $\zF^*:\A(E^*_2)\to\A(E^*_1)$ intertwines the corresponding de Rham differentials
\be\label{mor} \xd_{\zP_2}\circ\zF^*=\zF^*\circ\xd_{\zP_1}.
\ee
\end{definition}
\no Note that (local) sections $\za$ of $E_2^*$ are identified with (local) linear functions $\zi_\za$ on $E_2$, and $\zF^*(\za)$ corresponds to the linear function $\zF^*(\zi_\za)=\zi_\za\circ\zF$. If $(x^a,y^i)$ are affine coordinates in $\tilde U=\zt_2^{-1}(U)$, for a neighbourhood $U$ of a point $\zf(p)$, $p\in M_1$, then $\bar y^i=y^i\circ\zF$ are linear functions on $\zF^{-1}(\tilde U)$. Consequently, writing $\bar f=f\circ\zf$ for functions $f$ on $M_2$, the local condition (\ref{mor}) assuring that $\zF$ is a Lie algebroid morphism can be written in the form (cf. (\ref{drd}))
$$\xd_{\zP_1}\bar x^a=\bar\zr^a_i\,\bar y^i,\quad \xd_{\zP_1}\bar y^k=\half\bar c^k_{ji}\,\bar y^i\we \bar y^j\,,
$$
where $c_{ij}^k$ and $\zr^a_i$ are the structure coefficients of $\zP_2$. In other words, the above identities look formally like the structure equations (\ref{drd}) for $\zP_2$, with the only difference that the de Rham differential is the one for $\zP_1$. In what follows, it will be convenient to skip `bars' and view $(x^a,y^i)$ as functions on $E_1$, if $\zF$ is known and fixed.

\mn In particular, if $E_1=\sT\zS$ is the canonical Lie algebroid over a manifold $\zS$, then $\xd_{\zP_1}$ is the standard de Rham differential on differential forms, and the local condition assuring that $\zF:\sT\zS\to E=E_2$ is a Lie algebroid morphism reads
\be\label{moreq}
\xd x^a=\zr^a_i\,y^i,\quad \xd y^k=\half c^k_{ji}\,y^i\we y^j\,.
\ee
Here, of course, we interpret the linear functions $y^i$ on $\sT\zS$ as 1-forms on $\zS$.

\subsection{Poisson sigma models}
The Lie algebroid structure on $\sT^*M$, associated with a Poisson tensor $\zL$ on $M$, is usually described in terms of the \emph{Koszul bracket}
$$
[\za,\zb]_\zL=\Ll_{{\zL}^\#(\za)}\zb-\Ll_{{\zL}^\#(\zb)}\za-\xd\la\zL,\za\we\zb\ran\,,
$$
where $\zL^\#:\sT^*M\to\sT M$ is the VB-morphism induced by the Poisson tensor $\zL$.
In a more geometric approach, this Lie algebroid structure corresponds to the linear Poisson tensor $\dt\zL$ on $\sT M$, where $\dt\zL$ is the tangent lift of $\zL$  (cf. \cite{Grabowski:1995,Grabowski:1997,Grabowski:1997a}).
In fact, all tensor fields on a manifold $M$ can be lifted to $\sT M$ \cite{Yano:1973}, and the tangent lifts of multivector fields respect the Schouten-Nijenhuis bracket \cite{Grabowski:1995}.

If the Poisson tensor $\zL$ locally reads
$$\zL=\half\zL^{ij}(x)\pa_{x^i}\we\pa_{x^j},\quad\zL^{ij}=-\zL^{ji},
$$
then
\be\label{lift}
\dt\zL=\zL^{ij}(x)\pa_{x^i}\we\pa_{\dot x^j}+\half\frac{\pa \zL^{ij}}{\pa x^k}(x)\dot x^k\pa_{\dot x^i}\we\pa_{\dot x^j}.
\ee
In other words, the structure coefficients of the Lie algebroid $\sT^*M$ in the adapted coordinates $(x^i,\dot x^j)$ are
\[\zr^i_j=\zL^{ji},\quad c_{ij}^k=\half\zL^{ij}_{,k}.
\]
\no Hence, in the realm of Poisson sigma models, Lie algebroid morphisms $\zF:\sT\zS\to\sT^*M$ covering $\zf:\zS\to M$ are locally characterized by (cf. (\ref{moreq}))
\be\label{fe}
\xd X^i=\big(\zL^{ji}\circ\zf\big)\zh_j,\quad \xd \zh_k=\half\big(\zL^{ij}_{,k}\circ\zf\big)\zh_i\we \zh_j\,,
\ee
where $X^i=x^i\circ\zf$ and $\zh_j=p_j\circ\zF$, for Darboux coordinates $(x^i,p_j)$ on $\sT^*M$.

\mn When $\zS$ is a surface (with or without boundary), these are exactly the field equations for the Poisson sigma model associated with $\zL$, whose fields are VB-morphisms
\[
\xymatrix@C+30pt@R+10pt{
\sT\zS \ar[r]^{\zF}\ar[d]_{\zt_\zS} & \sT^*M\ar[d]^{\zp_M} \\
\zS \ar[r]^{\zf} & M\,,}
\]
and the action functional is
\[S(\zF)=\int_\zS \La\zF\,\overset{\we}{,}\,\Big(\sT\zf+\half\zL^\#\circ\zF\Big)\Ra,
\]
where the pairing is that between $\sT^*M$ and $\sT M$.
The action functional in local coordinates reads
$$S(X,\zh)=\int_\zS\Big(\zh_i\we\xd X^i+\half\zL^{ij}(X)\zh_i\we\zh_j\Big).
$$

\section{$G$-algebroids and their reductions}
If $\zt:P\to M=P/G$ is a principal bundle with the structure Lie group $G$ acting on $P$ by
$$h:G\ti P\to P, \quad h(g,p)=h_g(p),$$
then the $G$-action on $P$ can be lifted to a principal action $\hh$ on $\sT P$ by $\hh_g=\sT h_g$, and the corresponding principal bundle is
$$\hzt:\sT P\to\sA P:=\sT P/G.$$
We call $\hh$ the \emph{complete lift} of $h$. Note that $G$ acts on $\sT P$ \emph{via} $\hh$ by VB-automorphisms. The manifold $\sA P$ is canonically a vector bundle over $M$, and we have the following canonical morphism of vector bundles
\[
\xymatrix@C+30pt@R+10pt{
\sT P\ar[d]^{\zt_P}\ar[r]^{\hzt} & \sA P\ar[d]^{\hzt_P} \\
P\ar[r]^{\zt} & M\,.}
\]
Sections of $\sA P$ are canonically identified with $G$-invariant vector fields on $P$ which, clearly, are closed with respect to the Lie bracket. Therefore, the reduced vector bundle $\sA P$ inherits the structure of a Lie algebroid from $\sT P$, called traditionally the \emph{Atiyah algebroid} of the principal bundle $P$,
and $\hzt$ is a Lie algebroid morphism. The anchor map $\zr_P:\sA P\to\sT M$ is uniquely determined by the commutativity of the diagram
$$
\xymatrix@C+10pt@R+10pt{
\sT P\ar[rd]^{\sT\zt}\ar[rr]^{\hzt} && \sA P\ar[ld]^{\zr_P} \\
&\sT M &\,.}
$$
Note also that any morphism $\zc:P_1\to P_2$ of $G$-principal bundles, covering $\zc_0:M_1\to M_2$, induces canonically a Lie algebroid morphism of the corresponding Atiyah algebroids
$$
\xymatrix@C+40pt@R+20pt{
\sA P_1\ar[d]^{\hzt_{P_1}}\ar[r]^{\sA\psi} & \sA P_2\ar[d]^{\hzt_{P_2}}\\
 M_1\ar[r]^{\psi_0}& M_2}
$$
by the reduction of $\sT\zc$,
$$
\xymatrix@C+40pt@R+20pt{
\sT P_1\ar[d]^{\hzt_1}\ar[r]^{\sT\psi} & \sT P_2\ar[d]^{\hzt_2}\\
\sA P_1\ar[r]^{\sA\psi}& \sA P_2\,.}
$$
An abstract analog of the double bundle structures on $\sT P$ that we have just described is the following (cf. \cite[Definition 3.1]{Grabowski:2013} and \cite[Definition 5.36]{Grabowska:2021}).
\begin{definition}\label{def}  A principal $G$-bundle $\sp:E\to P$ with the $G$-action $\sh:G\ti E\to E$, endowed simultaneously with a vector bundle structure, is called a \emph{$G$-vector bundle} if $G$ acts \emph{via} vector bundle automorphisms. In the case when $E\to P$ is additionally a Lie algebroid, we call $E$ a $G$-algebroid if $G$ acts by Lie algebroid automorphisms.
In this case, the vector bundle (resp., Lie algebroid) structure and the $G$-principal bundle structure on $E$ are called \emph{compatible}.
\end{definition}
\no Note that if $(E,\sh)$ is a $G$-vector bundle, then the dual $E^*$ is also a $G$-vector bundle with the \emph{dual $G$-action} $\sh^*$ defined by $\sh^*_g=\big(\sh_{g^{-1}}\big)^*$. Its base manifold is the vector bundle $E^*/G=\sA^*P$, dual to $\sA P$.

\begin{remark}\label{rem1} The compatibility between the vector and principal bundle structures in Definition \ref{def} can be equivalently formulated (see \cite{Grabowski:2013}) as the commutation of the maps $\sh_g$ with the multiplications by reals in the vector bundle. This is based on the observation that any vector bundle structure is completely determined by the multiplication by reals (homotheties) \cite{Grabowski:2009}. In this sense, $E$ is a double, vector-principal bundle, called in \cite{Grabowska:2021} \emph{VB-principal bundle}. This immediately implies that the base $E_0=E/G$ of the principal bundle $\zt_0:E\to E_0$ carries an induced vector bundle structure $\sp_0:E_0\to M$, and $\sh$ induces a $G$-principal action $h$ on $P$, with the projection $\zt:P\to M$; the bases of the projections $\sp_0$ and $\zt$ can be canonically identified. Moreover, we have the commutative diagram
\[
\xymatrix@C+30pt@R+10pt{
E\ar[d]^{\sp}\ar[r]^{\zt_0} & E_0\ar[d]^{\sp_0} \\
P\ar[r]^{\zt} & M\,,}
\]
whose horizontal arrows define a morphism of vector bundles, and the vertical arrows a morphism of $G$-principal bundles. Of course, sections of $E_0$ can be viewed as $G$-invariant sections of $E$.
\end{remark}

\begin{theorem}\label{m1}\cite{Grabowska:2021,Grabowski:2013} Let $\sp:E\to M$ be a $G$-vector bundle. Then the map
\be\label{vRt} \Theta=\big(\zt_0,\sp\big):E\to E_0\ti_M P\ee
is a diffeomorphism identifying the vector and the $G$-principal bundle structures on $E$ with the product of the corresponding structures on $E_0\ti_M P$.

In the case when $E$ is additionally a Lie algebroid, $G$ acts by Lie algebroid automorphisms if and only if the Lie algebroid bracket $\lb$ is closed on $G$-invariant sections of $E\to P$, i.e., sections of $E_0$, and the anchor map $\zr:E\to\sT P$ is equivariant (cf. Remark \ref{rem1}),
$$\zr\circ\sh_g=\hh_g\circ\zr\quad\text{for all}\quad g\in G.$$
The latter means that $\zr(\ze)$ are invariant vector fields on $P$ for all invariant sections $\ze$ of $E\to P$.
As a result, we get a Lie algebroid structure on the reduced bundle $E_0$ -- the `Atiyah algebroid' of a $G$-algebroid.

Moreover, the anchor map $\zr:E\to\sT P$, being equivariant, goes down to a vector bundle morphism
$\zr^\#:E_0=E/G\to\sA P$,
\[
\xymatrix@C+30pt@R+10pt{
E\ar[d]^{\zt_0}\ar[r]^{\zr} & \sT P\ar[d]^{\hzt} \\
E_0\ar[r]^{\zr^\#} & \sA P\,,}
\]
covering the identity on $M$.
\end{theorem}
\no We call $E=E_0\ti_M P$ the \emph{canonical decomposition} of the $G$-vector bundle $E$. Of course, for any vector bundle $E_0\to M$ and any principal $G$-bundle $P\to M$ the fiber product $E=E_0\ti_M P$ is a $G$-vector bundle in an obvious way. It is a $G$-algebroid if and only if $E_0$ possesses a Lie algebroid structure with the anchor map $\zr_0:E_0\to\sT M$, and we have
a Lie algebroid morphism $\zr^\#:E_0\to \sA P$, covering the identity at the level of anchors,
$$
\xymatrix@C+30pt{
E_0\ar[d]^{\zr_0}\ar[r]^{\zr^\#} & \sA P\ar[d]^{\zr_P} \\
\sT M\ar@{=}[r] & \sT M\,.}
$$
In the above correspondence, we interpret sections $\ze$ of the vector bundle $E_0\to M$ as
$G$-invariant sections of the vector bundle $E\to P$, and $\zr^\#(\ze)$, being sections of $\sA P$, as invariant vector fields on $P$. If $\lb_0$ is the Lie bracket on sections of $E_0$, then the Lie bracket $\lb$ on sections of $E\to P$ is uniquely characterized by
$$[\ze_1,f\cdot\ze_2]=f\cdot[\ze_1,\ze_2]_0+\zr^\#(\ze_1)(f)\cdot\ze_2,$$
where $f$ is an arbitrary function on $P$. The Lie algebroid $(E_0,\lb_0)$ will be called the \emph{reduced Lie algebroid} of $(E,\lb)$.

\begin{definition} \emph{Morphisms} of $G$-algebroids are Lie algebroid morphisms which are simultaneously morphisms of the $G$-principal bundle structures.
\end{definition}
\no For such morphisms of $G$-algebroids, with the anchors $\zr_1,\zr_2$ we clearly have the following commutative diagram of Lie algebroid morphisms (cf. (\ref{mor}))
\be\xymatrix@R-8pt{
 & E^1 \ar[rrr]^{\Psi} \ar[dr]^{\zr_1}
 \ar[ddl]_{\zt^1_0}
 & & & E^2\ar[dr]^{\zr_2}\ar[ddl]_/-20pt/{\zt^2_0}
 & \\
 & & \sT P^1\ar[rrr]^/-20pt/{\sT\zc}\ar[ddl]_/-20pt/{\hzt^1}
 & & & \sT P^2 \ar[ddl]_{\hzt^2}\\
 E^1_0\ar[rrr]^/-20pt/{\Psi_0}\ar[dr]^{\zr_1^\#}
 & & & E^2_0\ar[dr]^{\zr_2^\#} & &  \\
 & \sA P^1\ar[rrr]^{\sA\zc}& & & \sA P^2 &
}\label{morr}
\ee
\no It is therefore easy to see the following.
\begin{proposition}\label{propmor} Let $E^i\to P^i$, $i=1,2$, be $G$-algebroids. Then a $G$-vector bundle morphism
$\Psi:E^1\to E^2$ is a morphism of $G$-algebroids if and only if, in the corresponding decompositions $E^i=E^i_0\ti_{M^i}P^i$, we have $\Psi=(\Psi_0,\psi)$, where the VB-morphism $\Psi_0:E^1_0\to E^2_0$ is a Lie algebroid morphism, and the VB-morphism $\sA\zc:\sA P^1\to\sA P^2$, induced by the morphism $\psi:P^1\to P^2$ of $G$-principal bundles, both covering the same map $\psi_0:M^1\to M^2$, is a Lie algebroid morphism which satisfies
\be\label{cond}\zr_2^\#\circ\Psi_0=\sA\psi\circ\zr_1^\#.\ee
\end{proposition}

\begin{example}\label{excan}
A \emph{canonical $G$-algebroid} over a manifold $M$ is $E=\sT M\ti G$, i.e., $E_0=\sT M$ and $P=M\ti G$, with the canonical Lie algebroid structure on $\sT M$, and the canonical Lie algebroid morphism
$\zr_E^\#:\sT M\to \sA P=\sT M\ti\g$ given by $\zr_E^\#(X)=(X,0)$, where $\g$ is the Lie algebra of the Lie group $G$. In particular, the anchor map $\zr_E:E\to\sT M\ti \sT G$ is $\zr_E(v,g)=(v,0_g)$.
\end{example}
\begin{example}\label{exp} With every $G$-principal bundle $\zt:P\to M$ there is canonically associated a $G$-algebroid, namely $E=\sT P$ with the complete lift $\hh$ of the $G$-action $\sh$ on $P$. In this case, the reduced Lie algebroid $E_0$ is the Atiyah algebroid of $P$, and $\zr^\#:\sA P\to\sA P$ is the identity.
\end{example}
\begin{remark} $G$-algebroids and $G$-groupoids have been introduced in \cite[Section 4]{Bruce:2017}. There is a structural result on $G$-groupoids \cite[Theorem 4.10]{Bruce:2017}, whose infinitesimal version is Theorem \ref{m1}.
\end{remark}
\no Since a Lie algebroid structure on $E$ corresponds to a linear Poisson structure $\zL$ on the dual bundle $\spi:E^*\to P$, it is easy to see that $G$ acts on $E$ by Lie algebroid automorphisms if and only if $G$ acts on $E^*$ \emph{via} the dual action $\sh^*$ by Poisson automorphisms. Such Poisson structures on $G$-principal bundles will be called \emph{$G$-Poisson structures}. In our case, $\zL$ is additionally linear, so it is a \emph{linear $G$-Poisson structure}. Thus, we have the following.
\begin{proposition}
There is a canonical one-to-one correspondence between $G$-algebroid structures on a $G$-vector bundle $E$ and
linear $G$-Poisson structures on $E^*$ with the dual $G$-action.
\end{proposition}
\begin{example}
The linear $G$-Poisson structure on the dual $E^*=\sT^*M\ti G$ of the \emph{canonical $G$-algebroid} $E=\sT M\ti G$ is $\zL_M\ti 0$ - the product of the canonical Poisson structure on $\sT^*M$ (the `inverse' of the canonical symplectic form) and the trivial Poisson structure on $G$.
\end{example}
\begin{example} The linear $G$-Poisson tensor on the dual $\sT^*P$ of the $G$-algebroid $\sT P$ with the $G$ action $\hh$ is the canonical Poisson structure $\zL_P$ on $\sT^*P$. It is automatically invariant with respect to the dual action $\hhs$.
\end{example}
\no In what follows, we put $G=\Rt$, as this is the Lie group closely related to Jacobi structures.
Hence, $\Rt$-algebroids will be called \emph{Jacobi algebroids}, a concept introduced in \cite{Grabowski:2001}.
\section{Line vs $\R^\ti$-principal bundles}
Vector bundles with one-dimensional fibers are called \emph{line bundles}. If $\zt: L\rightarrow M$ is a line bundle over a manifold $M$, then the submanifold $P=L^\times\subset L$ of non-zero vectors, $L^\times=L\setminus 0_M$, is canonically a principal bundle over $M$ with the structure group $(\R^\times, \cdot)$, i.e., the group of non-zero reals with multiplication, $(\Rt,\cdot)=\on{GL}(1;\R)$, and the projection $\zt:P\to M$ being the restriction of $\zt$. Of course, the $\R^\times$-action on $P=L^\times$ comes from the multiplication $h:\Rt\ti L\to L$ by non-zero reals in $L$.

Conversely, if we start with an $\Rt$-principal bundle $\zt: P\rightarrow M$, with the $\Rt$-action
$$\sh:\Rt\ti P\to P,\quad \sh(\zn,p)=\sh_\zn(p),$$
we can construct the line bundle $\zt:L_P\to M$ as the vector bundle associated with the standard $\Rt$-action on $\R$,
$$\R^\times\times\R\ni(\zn,r)\mapsto \zn r\in\R.$$
This means that elements of $L_P$ are the equivalence classes $[(p,r)]$ of pairs $(p,r)\in P\times \R$,
$$[(p,r)]=\{(\sh_\zn(p),r/\zn),\; \zn\in\R^\times\}.$$
A fundamental observation in this context is the following.
\begin{proposition}
There is a canonical principal bundle isomorphism between $L_P^{\times}$ and $P$. In particular, the line bundle $L_P$ is trivial if and only if the principal bundle $P$ is trivial.
\end{proposition}
\no The above isomorphism means that by gluing $L_P$ and $P$ out of local trivial bundles we can use the same transition functions. The dual bundle $\zt_*:L^\ast_P\to M$ can be characterized as the vector bundle associated with the opposite action of $\Rt$ on $\R$,
$$\R^\times\times \R\ni(\zn,r)\mapsto r/\zn\in\R.$$
Elements of $L^\ast_P$ are then equivalence classes of pairs $(p,q)\in P\times\R$, $$[(p,q)]=\{(h_\zn(p),\zn q),\; \zn\in\R^\times\},$$
and it is easy to see that each equivalence class $[(p,q)]$ corresponds to a homogeneous function of degree $1$ on the fibre $P_{\zt(p)}$. Consequently, (local) sections of $L^\ast_P$ correspond to (local) $1$-homogeneous functions on $P$. If $\zs$ is such a section, then the corresponding 1-homogeneous function will be denoted $\zi_\zs$. In what follows, to indicate the related line bundle, we will write $\Lst$ for $P$.

Note that in the case of a line bundle $L\to M$, every linear first-order differential operator $D:\Sec(L)\to\Sec(L)$ is automatically a VB-derivation, so $\sD L=(\sJ^1L)^*\ot_ML$. Moreover, the Lie algebroid $\sD L$ is exactly the Atiyah algebroid of the $\Rt$-principal bundle $\Lst=\big(L^*\big)^\ti$, i.e., $\sD L=\sA\Lst$, so we have the canonical identifications
\[ \sA\Lst=(\sJ^1L)^*\ot_ML=\sD L,
\]
and thus a canonical nondegenerate pairing
\[
\pair_L:\sJ^1L\ti_M\sD L\to L
\]
with values in the line bundle $L$.

The correspondence between line bundles and $\Rt$-bundles induces the concept of the \emph{category of line bundles} which, however, is not a full subcategory of vector bundles. Namely, as morphisms of $\Rt$-bundles are diffeomorphisms on fibers, morphisms in this category are those VB-morphisms which are linear isomorphisms on fibers. Therefore, not every VB-morphism of line bundles is a morphism in the category of line bundles. More precisely,
\begin{definition}\label{d1}
A \emph{morphism in the category of line bundles} is a VB-morphism $\zf:L^1\to L^2$ of line bundles, covering $\zf_0:M^1\to M^2$, such that
$$\zf_{x_0}:L^1_{x_0}\to L^2_{\zf_0(x_0)}$$
is a linear isomorphism for all $x_0\in M^1$.
\end{definition}
\no In particular, any section $\zs$ can be pulled back to a section $\zf^*(\zs)$ of $L^1$,
$$\zf^*(\zs)(x_0)=\zf_{x_0}^{-1}\Big(\zs\big(\zf_0(x_0)\big)\Big).$$
\begin{proposition}[\cite{Le:2018}] Any morphism $\zf:L^1\to L^2$ in the category of line bundles induces a morphism of Lie algebroids $\sD\zf:\sD L^1\to\sD L^2$, defined by
\[ (\sD\zf)(D)(\zs)=\zf\Big(D\big(\zf^*(\zs)\big)\Big).
\]
\end{proposition}

\section{Lifting $\Rt$-bundles}
As we already know, the $\Rt$-action $\sh:\Rt\ti\Lst\to\Lst$ can be lifted to an $\Rt$-principal action $\hh$ on $\sT\Lst$, and the dual action $\hhs$ on $\sT^*\Lst$. The base manifolds of these principal bundles are $\sD L$ and its dual $\sD^*L$, respectively, and the corresponding diagrams look like
\[
\xymatrix@C+30pt@R+10pt{
\sT \Lst\ar[d]^{\zt_\Lst}\ar[r]^{\hzt} & \sD L\ar[d]^{\hzt_L} \\
\Lst\ar[r]^{\zt} & M\,,}
\qquad\qquad\xymatrix@C+30pt@R+10pt{
\sT^* \Lst\ar[d]^{\zp_\Lst}\ar[r]^{\hts} & \sD^*L\ar[d]^{\hat\zp_L} \\
\Lst\ar[r]^{\zt} & M\,,}
\]
which leads to the identification
\be\label{A}
(\hzt,\zt_\Lst):\sT \Lst\to\sD L\ti_M\Lst,\quad\quad (\hts,\zp_\Lst):\sT^*\Lst\to\sD^*L\ti_M\Lst\,,
\ee
being particular cases of (\ref{vRt}) for $G=\Rt$.

To describe these actions in coordinates, let us start with local coordinates $(x^i)$ on an open subset $U$ of $M$, and coordinates $(x^i,s)$ in $\Lst$, associated with a local trivialization $\zt^{-1}(U)\simeq U\ti\Rt$, so $\sh_\zn(x^i,s)=(x^i,\zn\cdot s)$. We get then the adapted local coordinates $(x^i,s,\dot x^j,\dot s)$ in $\sT \Lst$, and the dual coordinates $(x^i,s,\zp_j,z)$ in $\sT^*\Lst$, in which
\be\label{sT}
\hh_{\zn}\big(x^i,s,\dot x^j,\dot s\big)=\big(x^i,\zn s,\dot x^j,\zn\dot s\big),\qquad
\hhs_\zn\big(x^i,s,\zp_j,z\big)=\big(x^i,\zn s,\zp_j,\zn^{-1}z\big)\,.
\ee
Therefore, in $\sD L$ we can use coordinates  $(x^i,\dot x^j,t=\dot s/s)$, and in $\sD^*L$ the coordinates
$(x^i,\zp_j,r=sz)$, so the projections $\hzt$ and $\hts$ read
\bea\label{dots}
&\hzt:\sT \Lst\to\sD L,\qquad (x^i,s,\dot x^j,\dot s\big)\longmapsto(x^i, \dot x^i, t={\dot s}/{s}),\\
& \hts:\sT^*\Lst\lra\sD^*L,\qquad \big(x^i,s,\zp_j,z\big)\longmapsto(x^i,\zp_j,r=sz)\,.\nn
\eea
An important observation is that in the case of the group $\Rt$, besides the tangent $\hzt$ and cotangent $\hts$ lifts of the principal action $\sh$, we have a whole family of such lifts and their duals (see, e.g., \cite{Bruce:2017,Grabowski:2013}). For instance, we can define $\Rt$-principal actions $\bh$ and $\hs$ on $\sT \Lst$ and $\sT^*\Lst$, respectively, by
\be\label{sT1}\bh_\zn=\zn^{-1}\cdot \sT \sh_\zn=\zn^{-1}\cdot \hh_\zn,\quad\qquad \h^*_\zn=\zn\cdot\big(\sT \sh_{\zn^{-1}}\big)^*=\zn\cdot\hhs_\zn.
\ee
These definitions work because $\Rt\subset\R$, and they have no direct analogs for a general Lie group $G$, because we cannot multiply vectors and covectors by group elements.

The lifted actions, $\bh$ and $\hs$, respect, like $\hh$ and $\hhs$, the canonical pairing between $\sT \Lst$ and $\sT^*\Lst$. The base manifolds of the principal bundles $(\sT^*\Lst,\h^*)$ and $(\sT \Lst,\h)$ are the first jet bundle $\sJ^1L^\ast_\Lst=\sJ^1L$ and its dual $\big(\sJ^1L^\ast_\Lst\big)^*=\big(\sJ^1L\big)^*$, respectively (cf. \cite{Bruce:2017,Grabowska:2024,Grabowski:2013}). The corresponding $\Rt$-vector bundles are

\[
\xymatrix@C+30pt@R+10pt{
\sT^* \Lst\ar[d]^{\zp_\Lst}\ar[r]^{\bt^*} & \sJ^1L\ar[d]^{\sj^1\zt} \\
\Lst\ar[r]^{\zt} & M\,,}
\qquad\qquad
\xymatrix@C+30pt@R+10pt{
\sT \Lst\ar[d]^{\zt_\Lst}\ar[r]^{\bt} & \big(\sJ^1L\big)^*\ar[d]^{(\sj^1\zt)^*} \\
\Lst\ar[r]^{\zt} & M\,,}
\]
which leads to another set of identifications (cf. (\ref{A}))
\[ (\bt,\zt_\Lst):\sT \Lst\to \big(\sJ^1L\big)^*\ti_M \Lst\quad\text{and}\quad
(\bt^*,\zp_\Lst):\sT^* \Lst\to \sJ^1L\ti_M \Lst,
\]
which refer to the principal as well as to the vector bundle structures.

\mn In local coordinates, we have
\[
\h_{\zn}\big(x^i,s,\dot x^j,\dot s\big)=\big(x^i,\zn s,\zn^{-1}\dot x^j,\dot s\big),\qquad
\h^*_\zn\big(x^i,s,\zp_j,z\big)=\big(x^i,\zn s,\zn\zp_j,z\big),
\]
and
\bea\label{J1L}
&\bt^*:\sT^\ast\Lst\lra\sJ^1L,\qquad (x^i,s,\zp_j,z)\longmapsto(x^i, p_j={\zp_j}/{s},z),\\
&\bt:\sT\Lst\lra\big(\sJ^1L\big)^*,\qquad (x^i,s,\dot x^j,\dot s)\longmapsto (x^i, \dot\x^j=s\dot x^j,\dot s).\nn
\eea
It is easy to see that the canonical pairing $\pair$ on $\sT^*\Lst\ti_\Lst\sT\Lst$ is not invariant with respect to the $\Rt$-action $\hs\ti\hh$, as it produces 1-homogeneous functions on fibers of $\Lst$, so elements in the fibers of $L$. In other words, we have a nondegenerate pairing of vector bundles
\[
\pair_L:\sJ^1L\ti_M\sD L\to L.
\]
In local coordinates associated with a local trivialization of $L$, this pairing reads
\[\La(x^i,\dot x^j,t),(x_i,p_k,z)\Ra_L=(x^i,\dot x^jp_j+tz).\]

\section{Jacobi bundles}\label{sec:Jacobi_bundles}
A local Lie bracket $J=\pb$ on sections of a line bundle $\zt:L\to M$ was introduced by Kirillov \cite{Kirillov:1976} under the name a \emph{local Lie algebra} and proven to be actually a bi-differential operator of the first order; for a purely algebraic generalization of this fact see \cite{Grabowski:1992}. For this reason, such brackets have been called in the literature \emph{Kirillov brackets} \cite{Bruce:2017, Grabowski:2013}. In the case of the trivial line bundle $L=M\ti\R$, the Kirillov bracket is defined on the space $\Ci(M)$ of functions on the manifold $M$ and commonly called a \emph{Jacobi bracket}. It is represented by a pair of tensor fields $J=(\zL,E)$ on $M$, where $\zL$ is a bi-vector field and $E$ is a vector field. The Jacobi bracket is then defined by
\[ \{f,g\}=\{f,g\}_J=\bk{\zL}{\xd f\we\xd g}+fE(g)-gE(f).
\]
In local coordinates,
\[\{f,g\}(x)=\zL^{ij}(x)\frac{\pa f}{\pa x^i}(x)\frac{\pa g}{\pa x^j}(x)+
{E^i}(x)\left(f(x)\frac{\pa g}{\pa x^i}(x)-\frac{\pa f}{\pa x^i}(x)g(x)\right),
\]
where $\zL^{ij}=-\zL^{ji}$. The Jacobi identity for the bracket can be expressed in terms of the Schouten-Nijenhuis  bracket $\lna \cdot,\cdot \rna_{SN}$ as
$$\lna E,{\zL}\rna_{SN}=0\,,\quad \lna{\zL},{\zL}\rna_{SN}=-2E\we{\zL}\,.
$$
\no In the full generality, i.e., for brackets on sections of line bundles, the corresponding structures $(L,\pb)$ were studied in detail by Marle \cite{Marle:1991} and called \emph{Jacobi bundles}. The Lie bracket $\pb$ is represented by an anti-symmetric  bi-differential operator of first order
$$J:\Sec(L)\ti\Sec(L)\to\Sec(L).$$
This implies that, for every section $\zs$ of $L$, the operator
$$D_\zs:\Sec(L)\to\Sec(L),\quad D_\zs(\zs')=\{\zs,\zs'\},$$
is a linear first-order differential operator, i.e., a section of the vector bundle $\sD L=(\sJ^1L)^*\ot_ML$.
We already know that such operators $D$ are VB-derivations in $L$: there is a vector field $\hat D$ such that
$$D(f\zs)-fD(\zs)=\hat D(f)\zs,$$
for any section $\zs$ of $L$ and any function $f$ on $M$. The vector field $\hat D$ is called the \emph{symbol} of $D$. The symbol map $\sb:\sD L\to\sT M$ extends to a short exact sequence of VB-morphisms covering the identity on $M$,
\[
0\lra M\ti\R\hlra\sD L\overset{\sb}{\lra}\sT M\lra 0.
\]

\mn Note that one can identify the Jacobi structure also with a vector bundle morphism $J:\we^2\sJ^1L\to L$ or, similarly to the Poisson case, a vector bundle morphism
\[
\xymatrix@C+25pt@R+20pt{
\sJ^1L \ar[d]^{\sj^1\zt}\ar[r]^{J^\#} & \sD L\ar[d]^{\hzt_\Lst}\\
M \ar@{=}[r] & M\,.}
\]
Therefore, with every section $\zs$ of $L$ we have associated the vector field
\begin{equation}
\label{Hvf}
X_\zs=\sb\big(J^\#(\,\sj^1(\zs)\big),
\end{equation}
called the \emph{Hamiltonian vector field} of $\zs$ \cite{Marle:1991}. The corresponding Jacobi-Hamiltonian mechanics is a rapidly developing subject, especially in the context of contact geometry  (see \cite{Grabowska:2022,Grabowska:2023,Grabowska:2024} and references therein). It is easy to see the following.
\begin{proposition} If in a chosen local trivialization of $L$ the Jacobi structure $J$ is represented by tensors $(\zL,E)$, then $J^\#:\sJ^1L\to\sD L=\sA\Lst $ takes in the adapted local coordinates the form
\be\label{redmap}J^\#(x^i,p_j,z)=\big(x^i,\zL^{kj}p_k+E^jz,-E^kp_k\big).\ee
\end{proposition}
\section{Poisson $\Rt$-bundles and Jacobi algebroids}
One can identify smooth sections $\zs$ of a line bundle $\zp:L\to M$ with smooth homogeneous functions $\zi_\zs$ of degree one on $\zt:L^*\to M$, and further also with homogeneous functions of degree one on the principal $\R^\ti$-bundle $\zt:\Lst\to M$, i.e., functions $f:L^\ti\to\R$ such that $f(\sh_s(v)):=f(s.v)=s\cdot f(v)$.

Given a Jacobi bracket $\pb$ on sections of $L$, we can try to define a Poisson bracket $\{\cdot,\cdot\}_\zP$, associated with a linear Poisson structure $\zP$ on $L^*$, using the identity
\begin{equation}\label{e1}
\iota_{\{\zs,\zs'\}} = \{\iota_{\zs}, \iota_{\zs'}\}_{\zP}.
\end{equation}
However, unlike the case of a Lie algebroid, this bracket is generally singular at points of the zero-section. Therefore, one has to define a Poisson tensor on the $\Rt$-principal bundle $P=\Lst=L^*\setminus 0_M$ instead. Indeed, in coordinates $(x^i,s)$ on $L^*$, dual to $(x^i,z)$ on $L$, we get
\begin{equation}\label{homP}
\zP(x,s)=\frac{1}{2s}\zL^{ij}(x)\pa_{x^i}\we\pa_{x^j}+{E^i}(x)\pa_s\we\pa_{x^i}\,.
\end{equation}
This identification allows for a very useful characterization of Jacobi brackets in terms of Poisson brackets (cf. \cite{Grabowski:2013, Marle:1991}).
\begin{definition}[\cite{Grabowski:2013}]
A \emph{Poisson $\R^\ti$-bundle} is a principal $\R^\ti$-bundle $(P,h)$ equipped with a Poisson  structure $\zP$ which is homogeneous of degree $-1$, i.e., $(\sh_s)_*(\zP)=s^{-1}\cdot\zP$. A \emph{morphism of Poisson $\Rt$-bundles} $\phi: P\rightarrow P'$ is a morphism of principal $\Rt$-bundles which is simultaneously a Poisson morphism.
\end{definition}
\no In what follows, a Poisson tensor on an $\Rt$-bundle $P$ which is homogeneous of degree $-1$ will be simply called \emph{homogeneous}. The corresponding Poisson bracket is closed on homogeneous (degree 1) functions. In particular, $\R^n\ti\R^\ti$ with coordinates $(x^i,s)$ and trivial $\Rt$-principal bundle structure, equipped with a Poisson tensor of the form (\ref{homP}), are basic examples of Poisson $\R^\ti$-bundles. Summarizing our observations we get the following,
\begin{proposition}
There is a one-to-one correspondence between Jacobi brackets $\pb$ on a line bundle $L\rightarrow M$ and homogeneous Poisson tensors $\zP$ on the principal bundle $P=\Lst$, given by (\ref{e1}).
\end{proposition}
\begin{example}
Let $(M,C)$ be a contact manifold, i.e., $C$ is a `maximally non-integrable' corank 1 distribution on $M$. The standard example is $C$ being the kernel of the canonical contact 1-form $\zh=\xd z-p_i\xd q^i$ on $\R^{2n+1}$ with coordinates $(q^i,p_j,z)$.
The maximal non-integrability can be described as the fact that the canonical symplectic form $\zw_M$ on $\sT^*M$ is still symplectic when restricted to $(C^o)^\ti$, where $C^o\subset\sT^*M$ is the annihilator of $C$. Since $C^o=L^*$ for the line bundle $L=\sT M/C$, this defines a Jacobi structure on $L$, corresponding to the homogeneous Poisson structure $\zP=\zw^{-1}$ on $\Lst$ being the inverse of the symplectic form $\zw=\zw_M\,\big|_{(C^o)^\ti}$. In other words, contact structures correspond to Poisson $\R^\ti$-bundles with invertible Poisson tensors \cite{Marle:1991,Bruce:2017,Grabowski:2013}. For $C=\ker(\zh)$ as above, we get the homogeneous Poisson tensor on the trivial $\Rt$-bundle $\Lst=\R^{2n+1}\ti\Rt$ of the form
$$\zP(q^i,p_j,z,s)=\pa_s\we\pa_z+\frac{1}{s}\pa_{p_i}\we\big(\pa_{q^i}+p_i\pa_z\big)\,.$$
This `homogeneous' approach to contact geometry is much simpler than the one dominating the literature and finds interesting applications in geometric mechanics \cite{Grabowska:2022,Grabowska:2024}.
\end{example}
\no Let us recall that by a Jacobi algebroid we mean an $\Rt$-algebroid (Definition \ref{def}), and $\sT\Lst$ is canonically a Jacobi algebroid (cf. \ref{exp}). The Lie algebroid structure on $\sT^*\Lst$, associated with the homogeneous Poisson tensor (\ref{homP}) on $\Lst$, corresponds to the tangent lift $\dt\zP$ of $\zP$ which lives on $\sT\Lst$.
Starting with bundle coordinates $(x^i,s)$ on $\Lst$ and using the adapted coordinates $(x^i,s,\pi_j,z)$ in $\sT^*\Lst$ and $(x^i,s,\dot x^j,\dot s)$ on $\sT\Lst$, we have, for $\zP$ as in (\ref{homP}),
\be\label{zPmap}\zP^\#:\sT^*\Lst\to\sT\Lst, \quad \zP^\#(x^i,s,\zp_j,z)=\big(x^i,s,s^{-1}\zL^{kj}\zp_k+E^jz, -E^k\zp_k\big).
\ee
Moreover, $\zP^\#:\sT^*\Lst\to\sT\Lst$ is a Lie algebroid morphism. Doing the reduction by $\Rt$-actions, $\hs$ on the left and $\hh$ on the right, we reconstruct $J^\#$. Indeed
the $\hs$ reduction in coordinates reads
$$(x^i,s,\zp_j,z)\mapsto(x^i,p_j=\zp_j/s,z),$$
and the $\hh$ reduction in coordinates reads
$$(x^i,s,\dot x^j,\dot s)\mapsto(x^i,\dot x^j,\dot s/s),$$
that easily implies that on the reduced bundles we get (\ref{redmap}).

\mn According to (\ref{lift}), in the adapted local coordinates $(x^i,s,\dot x^j,\dot s)$ on $\sT\Lst$, we have
\bea\label{Plift} \dt\zP&=&\Big(\frac{1}{2s}\frac{\pa \zL^{ij}}{\pa x^k}(x)\dot x^k-\frac{\dot s}{2s^2}\zL^{ij}(x)\Big)\pa_{\dot x^i}\we\pa_{\dot x^j}+\frac{1}{s}\zL^{ij}(x)\pa_{x^i}\we\pa_{\dot x^j}\\
&& +\frac{\pa E^i}{\pa x^k}(x)\dot x^k\pa_{\dot s}\we\pa_{\dot x^i}+E^i(x)\big(\pa_s\we\pa_{\dot x^i}+\pa_{\dot s}\we\pa_{x^i}\big).\nn
\eea
For the lifted action $\hh$ of the $\Rt$-action $\sh$ on $\Lst$, the coordinates $(x^i,s,\dot x^j,\dot s)$ are homogeneous of degrees $(0,1,0,1)$ (cf. \ref{sT}), so $\dt\zP$ is homogeneous and defines a Jacobi bracket on the line bundle `dual' to the $\Rt$-bundle $\sT\Lst\to\sT\Lst/\Rt=\sD L$, i.e., on the line bundle $\sD L\ti_M L\to\sD L$. This is the \emph{Jacobi lift} of the Jacobi structure defined in \cite{Grabowski:2001}.

\mn But for the lifted action $\h$, the coordinates $(x^i,s,\dot x^j,\dot s)$ are homogeneous of degrees $(0,1,1,0)$ (cf. \ref{sT1}), so $\dt\zP$ is $\Rt$-invariant. Hence, the dual $\Rt$-vector bundle is $\sT^*\Lst$ with the lifted $\Rt$-action $\hs$, and $\sT^*\Lst$ is a Jacobi algebroid with respect to this action, i.e., $\h^*_\zn$ are Lie algebroid automorphisms for all $\zn\in\Rt$. This, in turn, is the \emph{Poisson lift} of our Jacobi structure, as defined in \cite{Grabowski:2001}. Consequently, $\sJ^1L=\sT^*\Lst/\Rt$ inherits a Lie algebroid structure, and $\bt^*:\sT^*\Lst\to\sJ^1L$ is a Lie algebroid morphism.
This proves the following.
\begin{proposition}
Let $(L,J)$ be a Jacobi bundle, and let $\zP$ be the corresponding homogeneous Poisson tensor on the $\Rt$-principal bundle $\zt:\Lst\to M$. Then, the $\Rt$-bundle $(\sT^*\Lst,\hs)$ is canonically a Jacobi algebroid, and $\zP^\#:\sT^*\Lst\to\sT\Lst$ is a Jacobi algebroid morphism with respect to the $\Rt$-principal actions $\hs$ and $\hh$ on $\sT^*\Lst$ and $\sT\Lst$, respectively. In particular,
\be\label{commutat}\zP^\#\circ\h^*_\zn=\hh_\zn\circ\zP^\#.
\ee
The corresponding diagram consists of Lie algebroid morphisms and reads
\[\xymatrix@R-8pt@C-4pt{
 & \sT^*\Lst \ar[rrr]^{\zP^\#} \ar[dr]^{\bt^*}
 \ar[ddl]_{\zp_\Lst}
 & & & \sT\Lst\ar[dr]^{\hzt}\ar[ddl]_/-20pt/{\zt_\Lst}\ar[rr]^{\sT\zt}
 & &\sT M\ar@{=}[dr]&\\
 & & \sJ^1L\ar[rrr]^/-20pt/{J^\#}\ar[ddl]_/-20pt/{\sj^1\zt}
 & & & \sD L \ar[ddl]_{\hzt_L}\ar[rr]^{\sb}&&\sT M\ar[ddl]_{\zt_M}\\
 \Lst\ar@{=}[rrr]\ar[dr]^{\zt}
 & & & \Lst\ar[dr]^{\zt} & & && \\
 & M\ar@{=}[rrr]& & & M\ar@{=}[rr] & &M&\,.
}
\]

\mn In particular, the anchor map of the Lie algebroid $\sJ^1L$ is
\[\sb\circ J^\#:\sJ^1L\to\sT M.
\]
\end{proposition}
\no The Lie algebroid structure on $\sJ^1L$ we constructed here coincides locally with the Lie algebroid structure on $\sT^*M\ti\R$ obtained from the Jacobi structure $(\zL,E)$ on the trivial bundle $L=M\ti\R$ in \cite{Kerbat:1993}.
The authors there construct not the linear Poisson tensor (\ref{Plift}) but the Lie algebroid bracket $[(\za,f),(\zb,g)]$ on sections of the vector bundle $\sT^*M\ti\R$ by
\bea
[(\za,f),(\zb,g)]&=&\big(\Ll_{\zL^\#_\za}\zb-\Ll_{\zL^\#_\zb}\za-\xd<\zL,\za\we\zb>
+f\Ll_E\zb-g\Ll_E\za-i_E(\za\we\zb),\nn \\
&&\La\zL,\zb\we\za\Ra+\zL^\#_\za(g)-\zL^\#_\zb(f)+fE(g)-gE(f)\big). \label{J1-bracket}
\eea
Note that there is a simple characterization of the Lie algebroid bracket $\lb_{\zP^c}$ on $\sJ^1L$ (cf. \cite[Example 2.7]{Le:2018}).
\begin{proposition}
The Lie algebroid bracket $\lb_{\dt\zP}$ on $\sJ^1L$ is uniquely characterized as a Lie algebroid with the anchor $\zr\big(\,\sj^1(\zs)\big)=X_\zs$, where $X_\zs$ is the Hamiltonian vector field of $\zs$ (cf. (\ref{Hvf})), and the Lie bracket satisfying
$$[\,\sj^1(\zs),\sj^1(\zs')]=\sj^1\big(\{\zs,\zs'\}\big).$$
\end{proposition}

\begin{remark} An interesting description of Poisson and Jacobi structures in terms of the corresponding tangent lifts can be found in \cite{Grabowski:2003a,Grabowski:2001}. Note that the Lie algebroid structure on $\sJ^1L$ associated with a Jacobi structure on the line bundle $L$ is the `Jacobi analog' of the Lie algebroid structure on $\sT^*M$ associated with a Poisson structure on $M$.
\end{remark}
\section{Jacobi sigma models}
\subsection{The `constrained' approach}
Let us start with analyzing the Jacobi sigma model proposed in \cite{Bascone:2021,Bascone:2021a,Vitale:2023}. It is done for the trivial line bundle $L=M\ti\R$ and starts properly from a Poissonization (\ref{homP}) of a Jacobi structure $(\zL,E)$. After performing the standard Poisson sigma model procedure, and obtaining the field equations, the authors put \emph{ad hoc} one variable equal to zero and reduce the equations to fields, understood as vector bundle morphisms $\sT\zS\to\sT^*M\ti\R$. Locally, for Darboux coordinates $(x^i,p_j,z)$ on $\sT^*M\ti\R$, the fields are represented by 1-forms $(\zh_i,\z)$ on $\zS$, that correspond to linear coordinates $(p_i,z)$, and functions $X^i$ on $\zS$ corresponding to base coordinates $(x^i)$ on $M$. The main result of the paper is finding an action functional producing these field equations, namely
$$S(X,\zh,\z)=\int_\zS\Big(\zh_i\we\xd X^i+\half\zL^{ij}(X)\zh_i\we\zh_j-E^i(X)\zh_i\we\z\Big).
$$
It is clear that this action functional is not gauge invariant, i.e., it depends on the choice of a trivialization of $L$ (various trivializations are possible even for trivial bundles).
\subsection{The homogeneous Poisson approach}\label{HomPoi}
A natural approach is clearly \emph{via} Poisson sigma models applied to Poissonizations of given Jacobi bundles. For a Jacobi structure $J$ on the line bundle $\zt:L\to M$ consider the corresponding homogeneous Poisson structure $\zP$ on $\Lst$, which locally reads as in (\ref{homP}). The fields of the corresponding Poisson sigma model are vector bundle morphisms
\[
\xymatrix@C+30pt@R+10pt{
\sT\zS \ar[r]^{\zF}\ar[d]_{\zt_\zS} & \sT^*\Lst\ar[d]^{\zp_\Lst} \\
\zS \ar[r]^{\zf} & \Lst\,.}
\]
The Poisson tensor $\zP$ is represented by a vector bundle morphism
\[
\xymatrix@C+30pt@R+10pt{
\sT^*\Lst\ar[d]_{\zp_\Lst} \ar[r]^{\zP^\#} & \sT\Lst\ar[d]^{\zt_\Lst} \\
\Lst \ar@{=}[r] & \ \Lst\,,}
\]
and the action functional is the standard one,
\be\label{af}S(\zF)=\int_\zS \La\zF\,\overset{\we}{,}\,\Big(\sT\zf+\half\zP^\#\circ\zF\Big)\Ra.\ee
It is well known that solutions of the corresponding field equations are those $\zF$ which are Lie algebroid morphisms.

\mn The pull-backs by $\zf$ of local coordinates $(x^i,s)$ on $\Lst$ are functions $(X^i,\mathfrak{s})$ on $\zS$, $\s\ne 0$, and linear coordinates $(\pi_i,z)$ on $\sT^*\Lst$ represent 1-forms $(\bpi_i,\z)$ on $\zS$.
Hence, the action functional (\ref{af}) takes in coordinates the form (cf. (\ref{zPmap}))
\be\label{hpa}S(\zF)=\int_\zS\Big(\bpi_i\we\xd X^i+\z\we\xd\s+\frac{1}{2\s}\zL^{ij}(X)\bpi_i\we\bpi_j+E^j(X)\z\we\bpi_j\Big).
\ee
Let us notice that, apart from a coefficient multiplying $E$, this action functional has the same expression as the functional in the work \cite{Chatzistavrakidis:2020}, where, however, the authors considered only a trivializable line bundle and $\s$ was assumed to be a positive function. The Euler-Lagrange equations of motion can be written as follows:
\begin{align}\label{ELhomog}
& \xd X^i + \frac{1}{\s}\zL^{ij}(X)\bpi_j - \z E^i(X) = 0; \nonumber\\
& \xd \s + E^j(X)\bpi_j =0; \\
& \xd \bpi_k - \frac{1}{2\s}\frac{\partial \zL^{ij}}{\partial x^k}(X)\bpi_j \we \bpi_k - \frac{\partial E^j}{\partial x^k}(X)\z\we \bpi_j = 0; \nonumber \\
& \xd \z + \frac{1}{2\s^2}\zL^{ij}(X)\bpi_i \we \bpi_j = 0\nonumber \,,
\end{align}
showing that the morphism $\zF$ must be a Lie algebroid morphism (see \eqref{moreq}).
\begin{remark} A first impression is that the above Jacobi sigma model is nothing but a Poisson sigma model, and the fact that it is associated with a Jacobi structure is not visible in the action functional (\ref{af}). The point is that the fields are VB-morphisms into $\sT^*\Lst$, which is not just a Lie algebroid but a Jacobi algebroid, and $\Lst$ is not just a manifold but a principal bundle. This opens possibilities of defining equivalent variants of the Poisson sigma model in question, in which the Jacobi bundle sources can be explicitly seen.
\end{remark}

\mn Alternatively, one can obtain the equation of motion using the Hamiltonian formalism, as shown, for instance, in \cite{Cattaneo:2001}. Following this procedure, we consider a surface $\zS$ which is rectangle $\left[ 0,1 \right]\times \left[ -T,T \right] = I_u\times I_t$, with two coordinate functions $(u,t)$. The 1-forms $(\bpi_i,\z)$ can be written as $(\bpi^{(u)}_i\xd u + \bpi^{(t)}_i\xd t, \z^{(u)}\xd u + \z^{(t)}\xd t)$, with the boundary condition that $\bpi^{(t)}_i$ and $\z^{(t)}$ vanish on $\partial I_u$ (different boundary conditions have been investigated in \cite{Calvo:2006}). The action functional can be written as
\begin{align}\nn
&S(\p,\z,X,\s)=\int_{\zS}\Big( \bpi^{(u)}_i\partial_t X^i + \z^{(u)}\partial_t\s - \bpi^{(t)}_i\partial_u X^i - \z^{(t)} \partial_u\s \\
&+\frac{1}{\s}\zL^{ij}(X)\bpi^{(u)}_i\bpi^{(t)}_j + E^j(X)\big( \z^{(u)}\bpi^{(t)}_j - \z^{(t)}\bpi^{(u)}_j\big) \Big)\xd u\, \xd t\,.\nn
\end{align}
We interpret the model as a Hamiltonian system on a Banach manifold. The phase space is the space $\cP (\sT^*\Lst)\simeq \sT^*\cP (\Lst)$ of $C^1$ paths on $T^*\Lst$ such that the projected curves on $\Lst$ are of class $C^2$. This is a Banach manifold as explained in \cite{Crainic:2003}. The coefficients of the 1-form $\xd u$ in the first two equations in the system (\ref{ELhomog}), appear here as constraints and the forms $(\bpi^{(t)}_j, \z^{(t)})$ are Lagrange multipliers. From another point of view, the points of the phase space can be seen as morphisms of vector bundles $\eta = a\,\xd u\,\colon\, \sT I_u \,\rightarrow \, \sT^*\Lst$ so that the constraints appearing in the action functional mean that they are Lie algebroid morphisms. Interpreting these morphisms as curves, the curves that satisfy the constraints are called $A$-paths, following the nomenclature in \cite{Crainic:2003} and they form an embedded submanifold $\A (\Lst) \subset \sT^*\cP (\Lst)$.
\begin{definition}
Given a Lie algebroid $A$ with vector bundle projection $\pi$ and an interval $I_u=\left[ 0,1 \right]$, an $A$-path is a curve $a\,\colon\,I_u\,\rightarrow\,A$ on $A$ such that
\[
\frac{\xd}{\xd u}(\pi(a(u))) = \rho(a(u))\,,
\]
which implies that the vector bundle morphism $\eta = a \,\xd u \,\,\colon \, \sT I_u\,\rightarrow\, A$ is a Lie algebroid morphism.
\end{definition}
\mn The action functional contains the term,
\begin{equation*}
\int_0^1\Big( \bpi^{(u)}_i\pa_t X^i + \z^{(u)}\pa_t\s \Big) \xd u\,,
\end{equation*}
that we interpret as the tautological 1-form on $T^*\cP (\Lst)$, so the corresponding canonical symplectic structure $\omega_0$ provides the symplectic 2-form for this Hamiltonian description. Using a connection $\tilde{\nabla}$on $\Lst$ one can express the vectors tangent to the constrained submanifold $\A (\Lst)$ of $A$-paths in terms of a pair $(b, u)$ of a 1-form $b$ and a vector field $u$ along the $A$-path $a$ satisfying the following condition\cite{Crainic:2004}
\[
\zP^{\#}(b) = \zP^{\#}(\tilde{\nabla}_X a) - \left[ \zP^{\#}(a) , X \right]_{\Lst}\,,
\]
where $\left[ \cdot, \cdot \right]_{\Lst}$ denotes the bracket on the Lie algebroid $\sT^*\Lst$. The Hamiltonian function
\[
H_{(\bpi^t,\z^t)}= \int_0^1 \Big(\bpi^{(t)}_i \pa_u X^i - \z^{(t)} \pa_u\s +\frac{1}{\s}\zL^{ij}(X)\bpi^{(u)}_i \bpi^{(t)}_j + E^j(X)\big( \z^{(u)} \bpi^{(t)}_j - \z^{(t)} \bpi^{(u)}_j \big) \Big)\xd u
\]
defines a moment map for an action of the infinite-dimensional Lie algebra of 1-forms with the Lie algebra bracket $\left[ \cdot, \cdot \right]_{\Lst}$. Therefore, there is an associated integrable distribution $\cD \subset \sT\A (\Lst)$ which is generated by the tangent vectors $(b,u)\in \sT_a\A(\Lst)$ that can be written as
\[
u=\zP^{\#}(\beta)\,,\quad b=\nabla_{\zP^{\#}(\beta)} a + \left[ a, \beta \right]_{\Lst}\,.
\]
The quotient of the constrained submanifold by the action of the Lie algebra is a finite-dimensional topological manifold which is smooth if the Lie algebroid is integrable (see \cite{Crainic:2003} for the obstructions to the existence of a smooth structure on the quotient). It coincides with the manifold of equivalence classes of $A$-paths under the $A$-homotopy relation (see the definition below). In particular, this equivalence relation does not depend on the chosen connection on $\Lst$ \cite{Cattaneo:2001, Crainic:2004}.
\begin{definition}
A variation of $A$-paths is a family of $A$-paths $a_{t}\,\colon\,I_u\ti I_t\,\rightarrow \, A$ such that the projected family of curves $\pi(a_t)$ has fixed end-points. Alternatively, a variation $a_t$ of $A$-paths can be seen as a Lie algebroid morphism $a \xd u + \beta \xd t\,\colon\,\sT I_u\ti \sT I_t \,\rightarrow\, A$. We say that two paths $a_0$ and $a_1$ are $A$-homotopy equivalent if there is a variation of $A$-paths such that $a(u,0)=a_0$ and $a(u,1)=a_1$ with $\beta(0,t)=0=\beta(1,t)$.
\end{definition}
\no When the Poisson manifold $\Lst$ is integrable, the quotient manifold of $A$-paths under $A$-homotopies is a source simply-connected Lie groupoid $\cG$, with the symplectic structure $\omega_c$ obtained by reduction from the canonical symplectic form on the phase space. In the special case of an $\Rt$-algebroid the action $\hs$ of the group $\Rt$ on $\sT^*\Lst$ by Lie algebroid morphisms integrates to an action on the symplectic groupoid.

Indeed, for any $\nu\in \Rt$, if $a$ is an $A$-path, then the map $\hs \circ a\,\colon\,\sT I\,\rightarrow\,\sT^*\Lst$ defines another $A$-path. Moreover, if two $A$-paths $a_0$ and $a_1$ are $A$-homotopy equivalent, then the map $(\hs \circ a)\, \xd u + (\hs \circ \beta)\, \xd t$ is a $A$-homotopy between $\hs \circ a_0$ and $\hs \circ a_1$ because $\hs$ is a Lie algebroid homomorphism. Therefore, the symplectic groupoid $\cG$ is endowed with an $\Rt$-action $\mathfrak{h}\,\colon\, \Rt\times \cG\,\rightarrow\,\cG$ and it follows straightforwardly from the operations on the symplectic groupoid that the maps $\mathfrak{h}_{\nu}$ are groupoid homomorphisms, so that $\cG$ is an $\Rt$-groupoid. This result is in agreement with \cite[Theorem 4.7]{Bruce:2017}.

Eventually, since the canonical symplectic structure $\omega_{\Lst}$ on $\sT^*\Lst$ is homogeneous of degree 1, \cite[Proposition 8]{Grabowska:2023}, also the symplectic form $\omega_0$ on the phase space $\sT^*\cP(\Lst)$ is homogeneous of degree 1 and, eventually, the symplectic form $\omega_c$ on $\cG$ is also homogeneous of degree 1. Indeed, the symplectic structure on $\cG$ is obtained as follows:
\[
\omega_c\big(\hat{a})(\hat{\xi}_1,\hat{\xi}_2\big) = \int_0^1 \omega_{\Lst}\big(a(u)\big)\big(\xi_1(u),\xi_2(u)\big) \xd u\,,
\]
where $\hat{a}$ represents an equivalence class of $A$-paths, $\hat{\xi}_i$, with $j=1,2$, are equivalence classes of vectors, whereas, the corresponding expressions without the hat are generic elements of the equivalence class. Therefore, we have
\[
\begin{split}
\omega_c\big(\mathfrak{h}_{\nu}(\hat{a})\big)\big(\sT(\mathfrak{h}_{\nu})(\hat{\xi}_1),
\sT(\mathfrak{h}_{\nu})(\hat{\xi}_2)\big) &= \int_0^1 \omega_{\Lst}\Big(\hs_{\nu}\big(a(u)\big)\Big)\Big(\sT(\hs_{\nu})\big(\xi_1(u)\big),
\sT(\hs_{\nu})\big(\xi_2(u)\big)\Big) \xd u \\
&= \nu \int_0^1 \omega_{\Lst}\big(a(u)\big)\big(\xi_1(u),\xi_2(u)\big) \xd u\,.
\end{split}
\]
In the first equality, we have used the definition of the action on the quotient, whereas, in the second equality, we have used the homogeneity of the canonical symplectic structure on $\sT^*\Lst$. Summarizing these observations, we get the following.
\begin{theorem}
If $\Lst$ is an integrable Poisson manifold, the integrating symplectic groupoid $\cG$ is an $\Rt$-groupoid with a contact structure, i.e., it is endowed with an $\Rt$-action $\mathfrak{h}\,\colon\, \Rt\times \cG\,\rightarrow\,\cG$, where the maps $\mathfrak{h}_{\nu}$ are groupoid morphisms, and the symplectic structure is homogeneous of degree 1 with respect to this action.
\end{theorem}
\no This means that the groupoid $\cG$ integrating the Lie algebroid $\sT^*\Lst$ endowed with the 2-form $\omega_c$ is a contact groupoid in the sense of \cite{Bruce:2017}.

\begin{example}\label{ex_1}
As a simple example of the previous construction, we can consider the trivial Jacobi bundle $L^{\ti}=\R^{2k+1} \times \Rt\simeq \Lst$ and $\zS=I_u\times I_t$. As we have already noticed, in this case, we have global coordinate functions $(x^i , s)$, with $i=0,\cdots ,2k$ on $\Lst$, whereas on $\zS$ we have coordinate functions $(u, t)$. Since the Jacobi bundle is trivial, as already explained in Section \ref{sec:Jacobi_bundles}, the pair ($\zL$, $E$) represents the Jacobi bracket globally. In particular, the bracket between two homogeneous functions of degree 1 will be
\[
\left\lbrace f_1s, f_2s\right\rbrace = \Big(\zL(\xd f_1, \xd f_2) + f_1 E(f_2)  - f_2 E(f_1)\Big)\cdot s = \left\lbrace f_1, f_2\right\rbrace_J \cdot s \,.
\]
Let us consider the case of a homogeneous Poisson structure that is invertible, so that it can be expressed in terms of the following tensor fields:
\[
\zL = \sum_{j=1}^k \Big(\partial_{x^j} - x^{k+j}\partial_{x^0} \Big) \we \partial_{x^{k+j}}\,,\quad E = \partial_{x^0}\,.
\]
The fields of the Jacobi sigma model are the functions $(X^i , \s)$ and the 1-forms
$$\bpi^{(u)}_i \xd u + \bpi^{(t)}_i \xd t\,,\ \z^{(u)} \xd u + \z^{(t)} \xd t$$
on $\zS$. Then, the action functional for the homogeneous approach can be written as
\[
\begin{split}
S(X^i , \s, \bpi  , \z) = \int_\zS\Big(\bpi^{(u)}_i\pa_t X^i &+\z^{(u)} \pa_u\s - \bpi^{(t)}_i \pa_u X^i - \z^{(t)} \pa_u\s + \\
&+\frac{1}{\s}\zL^{ij}(X)\bpi^{(u)}_i \bpi^{(t)}_j+E^j(X)\big(\z^{(u)} \bpi^{(t)}_j - \z^{(t)} \bpi^{(u)}_j\big) \Big)\xd u \, \xd t\,,
\end{split}
\]
from which one derives the equations of motion in Eq.\eqref{ELhomog}. Since the Poisson bracket is invertible, the symplectic groupoid corresponds to the monodromy groupoid of $\Lst$, which is the space of equivalence classes of paths in $\Lst$ under the standard homotopy equivalence relation. The quotient of $T^*{\Lst}$ under the $\Rt$ action is identified with the bundle $\sJ^1 L\simeq\sT^*\R^{2k+1}\times \R$ endowed with the Lie algebroid bracket in (\ref{J1-bracket}). Introducing the reduced data $(X^j,\s)$ and $(\eta_i = \bpi_i/\s,\z)$, the $A$-paths on $\sT^*\Lst$ satisfy the equations
\beas
&&\partial_{u}(X^0) = \z\; \\
&&{\partial_u}(X^j) = \zL^{jl}\eta_l\,,\quad j=1\,,\cdots\,,2k\,;\\
&&{\partial_u}\big(\log(s)\big) = E^j\eta_j = \eta_0\,,
\eeas
so that one recognizes that the first two sets of equations determine an $A$-path on the Lie algebroid $\sT^*\R^{2k+1}\ti \R$, whereas the last equation can be integrated to obtain
\[
s(1) = s(0)\exp\Big( \int_0^1 \big(E^j(X(u))\,\eta_j(X(u))\big)\xd u \Big)\,.
\]
As shown in \cite{Crainic:2006}, the integral on the right-hand side is a function of an $A$-path on $\sJ^1L$ which depends only on its homotopy class and can be identified with a multiplicative function of the groupoid $\cG_c$ integrating $\sJ^1L$, which means:
\[
\int_0^1 \big(E^j(X(u))\,\eta_j(X(u))\big)\xd u = f(\left[a\right])\,\quad \mathrm{and}\; f(\left[a_2 \circ a_1 \right]) = f(\left[a_2\right]) + f(\left[a_1\right])\,,
\]
where $\left[ a \right]$ denotes a homotopy class of $A$-paths on $\sJ^1 L$.

Therefore, the groupoid $\cG$ of equivalence classes of $A$-paths on $\sT^*\Lst$ under homotopies can be identified with the manifold $\cG \simeq \Rt\ti \R_+ \ti \R^{2k+1}\ti \R^{2k+1}$ which is a groupoid with the following source and target maps $\alpha\,,\beta\,\colon\,\cG \,\rightarrow \, \Lst$
\[
\alpha(\s,t , X_l^j, X_r^j)= (s, X^j_l)\,,\quad \beta(\s,t,X^j_l,X^j_r)=(t\s, X^j_r)\,,
\]
and the multiplication map defined by
\[
(t_1\s_1 , t_2 , X_r^j, Y^j)\circ (\s_1,t_1 , X_l^j, X_r^j) = (s_1, t_2 t_1, X^j_l, Y^j)\,.
\]
The action of the group $\Rt$ by the map $\hs_{\nu}$ on the Lie algebroid gives rise to the following action, say $\mathfrak{h}_{\nu}$, on the groupoid $\cG$:
\[
\mathfrak{h}_{\nu} (\s,t , X_l^j, X_r^j) = (\nu \s, t, X_l^j, X_r^j)\,,
\]
and one can check straightforwardly that it is an action via groupoid morphisms. Consequently, the quotient groupoid is identified with the groupoid $\cG_c = \R_+ \ti \R^{2k+1}\ti \R^{2k+1}$. The space $\Lst$ is endowed with a homogeneous symplectic structure $\omega_0 = - s \xd \theta + \xd s \we \theta$ where $\theta = \xd x^0 - \sum_{j=1}^{k}x^{k+j}\xd x^j$, so that, the symplectic groupoid $\cG$ will be endowed with the multiplicative symplectic structure
\[
\omega = \alpha^*(\omega_0) - \beta^*(\omega_0) = s\cdot\xd\big(\alpha^*(\theta)\big) + \xd s \we \alpha^*(\theta) - \xd(ts) \we \beta^*(\theta) + st\cdot \xd\big(\beta^*(\theta)\big)\,,
\]
which is homogeneous of degree 1. This symplectic structure is associated with a trivial $\Rt$-principal bundle which is a submanifold of $\sT^*\cG_c$. In particular, it can be identified with a contact 1-form $\theta_c$ on $\cG_c$ which is written as
\[
\theta_c = \xd X_l^0 - \sum_{j=1}^{k}X_l^{k+j}\xd X_l^j - t \Big( \xd X_r^0 - \sum_{j=1}^{k}X_r^{k+j}\xd X_r^j \Big)\,.
\]
Therefore, the groupoid $\cG_c$ is the groupoid that integrates the Lie algebroid $\sJ^1 L$ and it is endowed with a global contact 1-form. This reduced groupoid is a contact groupoid in the terminology used by some authors \cite{Crainic:2006}, but the corresponding definition is substantially more complicated than the `homogeneous' one in \cite{Bruce:2017}.
\end{example}

\subsection{Jacobi algebroid approach}
To put some `Jacobi flavour' to the above homogeneous Poisson sigma model, let us use the fact that $\sT^*\Lst$ is not only a Lie algebroid but a Jacobi algebroid. One could expect that taking the homogeneity into account, our fields should not be just vector bundle morphisms $\zF$, but morphisms of $\Rt$-vector bundles, with the natural domain being $E=\sT\zS\ti\Rt$ -- the canonical Jacobi algebroid for a manifold $\zS$, associated with the trivial $\Rt$-principal bundle $P=\zS\ti\Rt$ and the reduced Lie algebroid $E_0=\sT\zS$.
\be\label{Psi}
\xymatrix@C+30pt@R+10pt{
\sT\zS\ti\Rt \ar[r]^{\Psi}\ar[d]_{\zt_P} & \sT^*\Lst\ar[d]^{\zp_\Lst} \\
P=\zS\ti\Rt \ar[r]^{\zc} & \Lst\,.}
\ee
In the above, $\zc:\zS\ti\Rt\to \Lst$ is a morphism of $\Rt$-principal bundles, covering a map $\zf_0:\zS\to M$.
But the $\Rt$-vector bundle morphism (\ref{Psi}) is completely determined by a VB-morphism
\be\label{Phi}
\xymatrix@C+30pt@R+10pt{
\sT\zS \ar[r]^{\zF}\ar[d]_{\zt_\zS} & \sT^*\Lst\ar[d]^{\zp_\Lst} \\
\zS \ar[r]^{\zf} & \Lst\,,}
\ee

\no where $\zF(v)=\Psi(v,1)$, $\zf(\zg)=\zc(\zg,1)$, and $\zf:\zS\to\Lst$ covers $\zf_0$. In the decomposition $\sT^*\Lst=\sJ^1L\ti_M\Lst$, we have
$$\zF(v)=\big(\zF_0(v),\zf(v)\big),$$
where
$$\zF_0:\sT\zS\to\sJ^1L,\quad \zF_0=\bt^*\circ\zF,$$
is a VB-morphism covering $\zf_0:\zS\to M$.

\mn Conversely, any morphism (\ref{Phi}) induces a morphism (\ref{Psi}) of $\Rt$-vector bundles, defined by
\be\label{PF}\Psi(v,s)=\h^*_s\big(\zF(v)\big)=\Big(\zF_0(v),\sh_s\big(\zf(v)\big)\Big).
\ee
It covers the map
$$\zc:\zS\ti\Rt\to\Lst,\quad \zc(\zg,s)=\sh_s\big(\zf(\zg)\big).$$
This shows that the fields $\Psi$ (\ref{Psi}) are in a one-to-one correspondence with the fields $\zF$ (\ref{Phi}), which we have used as the fields in the homogeneous approach. A fundamental observation in this context is the following.
\begin{theorem}\label{the}
The VB-morphism (\ref{Phi}) is a Lie algebroid morphism if and only if (\ref{Psi}) is a morphism of Jacobi algebroids.
\end{theorem}
\begin{proof} The `only if' part is trivial, so suppose $\zF$ is a Lie algebroid morphism. As $\bt^*$ is a Lie algebroid morphism, also $\zF_0:\sT\zS\to\sJ^1L$ is a Lie algebroid morphism. In view of Proposition \ref{propmor}, it suffices to prove that the anchor maps intertwine $\Psi$ with $\sT\psi$, i.e., the diagram
\be\label{Psi1}
\xymatrix@C+40pt@R+20pt{
E=\sT\zS\ti\Rt \ar[r]^{\Psi}\ar[d]_{\zr_E} & \sT^*\Lst\ar[d]^{\zP^\#} \\
\sT P=\sT\zS\ti\sT\Rt \ar[r]^/15pt/{\sT\zc} & \sT\Lst}
\ee
is commutative. But the anchor map $\zr_E$ is (see Example \ref{excan}) $\zr_E(v,s)=(v,0_s)$ and
$$\sT\zc(v,0_s)=\sT\sh_s\big(\sT\zf(v)\big)=\hh_s\big(\sT\zf(v)\big).$$
On the other hand (see (\ref{PF})),
$$\zP^\#\circ\Psi(v,s)=(\zP^\#\circ\h^*_s)(\zF(v)).$$
But (cf. \ref{commutat}) $\zP^\#\circ\h^*_s=\hh_s\circ\zP^\#$,
so
$$\zP^\#\circ\Psi(v,s)=(\hh_s\circ\zP^\#)(\zF(v)).$$
Hence, the commutativity of (\ref{Psi1}) is equivalent to
$$\hh_s\circ\sT\zf=\hh_s\circ\zP^\#\circ\zF$$
and, consequently, to $\sT\zf=\zP^\#\circ\zF$. But this is satisfied, as $\zF$ is a Lie algebroid morphism.

\end{proof}
\no The above theorem shows that the Jacobi algebroid version of Jacobi sigma-models is equivalent to the homogeneous version if we use the same action functional.
\begin{definition}[alternative]
A \emph{Jacobi sigma model} associated with a Jacobi structure $J$ on a line bundle $\zp:L\to M$ is the sigma model in which fields are $\Rt$-vector bundle morphisms (\ref{Psi}) and the action functional is
\[S\big(\Psi\big)=\int_\zS \La\zF\,\overset{\we}{,}\,\Big(\sT\zf+\half\zP^\#\circ\zF\Big)\Ra,
\]
where $\zF=\Psi(\cdot,1):\sT\zS\to\sT^*\Lst$, $\zf=\zc(\cdot,1):\zS\to\Lst$, and $\zP$ is the homogeneous Poisson tensor on $\Lst$ associated with $J$ (cf. (\ref{homP})). It is exactly like (\ref{af}), so Theorem \ref{the} immediately implies the following.
\end{definition}
\begin{theorem}
Solutions of the field equations of this Jacobi sigma model are those $\Rt$-vector bundle morphisms which are simultaneously Jacobi algebroid morphisms.
\end{theorem}
\no Note that the fact that $\Psi$ is a Jacobi algebroid morphism means that the following commutative diagram, which is just (\ref{morr}) in the present case, consists of Lie algebroid morphisms.
\be\xymatrix@R-2pt{
 & \sT\zS\ti\Rt \ar[rrr]^{\Psi} \ar[dr]^{\zr}
 \ar[ddl]_{\zt_1}
 & & & \sT^*\Lst\ar[dr]^{\zP^\#}\ar[ddl]_/-10pt/{\bt^*}
 & \\
 & & \sT\zS\ti\sT\Rt\ar[rrr]^/-20pt/{\sT\zc}\ar[ddl]_/-20pt/{\hzt_1}
 & & & \sT\Lst \ar[ddl]_{\hzt}\\
 \sT\zS\ar[rrr]^/+30pt/{\zF_0}\ar[dr]^{\on{id}\ti 0}
 & & & \sJ^1L\ar[dr]^{J^\#} & &  \\
 & \sT\zS\ti\R\ar[rrr]^{\sD\zf}& & & \sD L &
}\label{morr1}
\ee
Here, $\zr(v,s)=(v,0_s)$, $\zt_1(v,s)=v$, and $\hzt_1\big(v,(s,r)\big)=(v,r)$. The condition (\ref{cond}) takes the form
\be\label{cond1}
\sD_0\zf := \sD\zf\circ (\on{id}\ti 0) = J^\#\circ\zF_0\,.
\ee

\subsection{The reduced model}
If one wants to use the original data and not the homogeneous Poisson picture for a Jacobi structure, the most appropriate approach seems to be a model referring only to the latter. We already know that the manifold $M$ with its Poisson structure and the Lie algebroids $\sT^*M$ and $\sT M$, being the main geometrical objects in the Poisson sigma model, are replaced in the Jacobi picture by the line bundle $L\to M$, the Lie algebroid $\sJ^1L$, and the gauge algebroid $\sD L$. Consequently, the algebra $\Ci(M)$ of functions on $M$ is replaced by the $\Ci(M)$-module $\Sec(L)$, differentials of functions are replaced by jets of sections of $L$, and vector fields acting on functions are replaced by first-order differential operators action on sections of $L$. Like a Poisson structure on $M$ is encoded in a Lie algebroid morphism $\zL^\#:\sT^*M\to\sT M$, the Jacobi bundle structure on $L$ is encoded in the Lie algebroid morphism $J^\#:\sJ^1L\to\sD L$.

\mn Note that the smooth map $\zf:\zS\to\Lst\subset L^*$ covering $\zf_0:\zS\to M$, present in the homogeneous Poisson picture, is completely equivalent to a morphism $\phi:\zS\ti\R\to L$ in the category of line bundles (see Definition \ref{d1}), covering $\phi_0=\zf_0$. Indeed, $\zf$ and $\phi$ are related by
$$\Bk{\phi(\zg,1)}{\zf(\zg)}=1.$$
Alternatively, $\phi$ defines a trivialization of the pullback bundle $\phi_0^*L$, so we can identify sections of this pullback bundle with functions on $\zS$. More precisely, in pairwise dual local coordinates $(x^i,z)$ and $(x^i,s)$ on $L$ and $L^*$, respectively, the function $f(z)$ on $\zS$, associated with the local section $z=z\big(\phi_0(\zg)\big)$ of the pullback bundle $\phi_0^*L$, is given by
\be\label{funct}
f(z)=\s\cdot z,
\ee
where $\s=s\circ\zf$.

\mn Our fields will be pairs $(\zF_0,\phi)$, where $\zF_0:\sT\zS\to\sJ^1L$ and $\phi:\zS\ti\R\to L$ are VB-morphisms, the second assumed to be regular (i.e., a morphism in the category of line bundles), both covering the same map $\phi_0:\zS\to M$. Alternatively, the fields are VB-morphisms $(\zF_0,\zf)$ of $\sT\zS$ into the vector bundle $\sJ^1L\ti_M\Lst$ over $\Lst$.
\mn As we already know, the morphism $\phi$ induces a Lie algebroid morphism
\[\sD\phi:\sD(\zS\ti\R)=\sT\zS\ti\R\to\sD L
\]
which, reduced to $\sT\zS$, yields a Lie algebroid morphism
\[\sD_0\phi:\sT\zS\to\sD L.
\]
It is easy to see that $\sD_0\phi= \hzt\circ\sT\zf$,
\be\label{dd1}
\xymatrix@C+25pt{
\sT\zS \ar[r]^{\sT\zf}\ar[d]_{\zt_\zS}\ar@/^2pc/[rr]^{\sD_0\phi} & \sT\Lst\ar[d]^{\zt_\Lst}\ar[r]^{\hzt} & \sD L\ar[d]^{\hzt_L} \\
\zS \ar[r]^{\zf}\ar@/_2pc/[rr]_{\phi_0} & \Lst\ar[r]^{\zt} & M\,.}
\ee
Another VB-morphism of this type is
$$J^\#\circ\zF_0:\sT\zS\to\sD L,$$
and the corresponding diagram reads
\be\label{dd2}
\xymatrix@C+30pt@R+10pt{
\sT\zS \ar[r]^{\zF_0}\ar[d]_{\zt_\zS}\ar@/^2pc/[rr]^{J^\#\circ\zF_0} & \sJ^1L\ar[d]^{\sj^1\zt}\ar[r]^{J^\#} & \sD L\ar[d]^{\hzt_L} \\
\zS \ar[r]^{\phi_0}\ar@/_2pc/[rr]_{\phi_0} & M \ar@{=}[r] & M\,.}
\ee
Our action functional is now
\be\label{afr}
S(\zF_0,\phi)=\int_\zS \La\zF_0\,\overset{\we}{,}\,\Big(\sD_0\phi+\half J^\#\circ\zF_0\Big)\Ra_L.
\ee
The reader deserves an explanation. The canonical pairing $\pair_L:\sJ^1L\ti_M\sD L\to L$ takes values in $L$, but what we integrate is a two-form on $\zS$ with values in the pullback bundle $\phi_0^*L$, which is trivialized by $\phi$, so we can actually view it as a standard two-form on $\zS$.

Let us see what this two-form looks like in local coordinates. Given affine coordinates $(x^i,z)$ on the line bundle $L$, and the dual coordinates $(x^i,s)$ on the dual bundle, we have the adapted coordinates $(x^i,p_j,z)$ on $\sJ^1L$, and $(x^i,\dot x^j,t)$ on $\sD L$. With these coordinates, the map $\zf:\zS\to\Lst$ associates local functions $X^i=x^i\circ\phi_0$ and $\s=s\circ\zf$, and the VB-morphism $\zF_0$ associates local 1-forms $\p_j,\z$ on $\zS$.

The coordinate form (\ref{dots}) of the map $\hzt:\sT\Lst\to\sD L$ implies that, in turn, with $\sD_0\phi$ there are associated 1-forms $\xd X^i$ and $\xd\s/\s$, and
\be\label{1pair}
\La\zF_0\,\overset{\we}{,}\,\sD_0\phi\Ra_L=\p_i\we\xd X^i+\frac{1}{\s}\z\we\xd\s.
\ee
According to the form (\ref{redmap}) of $J^\#$, we have also
\be\label{2pair}
\La\zF_0\,\overset{\we}{,}\,\half J^\#\circ\zF_0\Ra_L=\half\zL^{ij}(X)\p_i\we\p_j+E^j(X)\z\we\p_j.
\ee
Both (\ref{1pair}) and (\ref{2pair}) are two-forms on $\zS$ with values in the pullback bundle $\phi_0^*L$, so they represent actual two-forms on $\zS$ after being multiplied by $\s$ (cf. (\ref{funct})). Consequently,
\[ S(\zF_0,\phi)=\int_\zS\Big(\s\p_i\we\xd X^i+\z\we\xd\s+\frac{\s}{2}\zL^{ij}(X)\p_i\we\p_j+\s E^j(X)\z\we\p_j\Big).
\]
Let us observe now that, according to (\ref{J1L}), $(x^i,s,\zp_j=sp_j,z)$ are the corresponding adapted coordinates on $\sT^*\Lst$, so our action functional is actually the same as the action functional (\ref{hpa}) in the homogeneous Poisson model, in which we just replace $\bpi_i$ with $\s\p_i$.
Solutions are then Lie algebroid morphisms $\zF:\sT\zS\to\sT^*\Lst$ which, according to Theorem \ref{the}, are the same as the corresponding morphisms $\Psi:\sT\zS\ti\Rt\to\sT^*\Lst$ of Jacobi algebroids. In other words, the diagram (\ref{morr1}) consists of Lie algebroid morphisms. The commutativity requires (\ref{cond1}), that reduces to
$\sJ^\#\circ\zF_0=\sD_0\phi$. This means that the side parts of diagrams (\ref{dd1}) and (\ref{dd2}) are the same.
Summing up, we get the following.

\begin{theorem}
Solutions induced by the action functional (\ref{afr}) are fields $(\zF_0,\phi)$ such that $\zF_0:\sT\zS\to\sJ^1L$ is a Lie algebroid morphism and $J^\#\circ\zF_0=\sD_0\phi$.
\end{theorem}

\begin{example}
Let us consider as the source manifold of the Jacobi sigma model $ \Sigma = I_u\times I_v$, like in Example \ref{ex_1}, whereas as target manifold we take the M\"obius band $L$ which is a nontrivial line bundle over the circle $\tau\,\colon\,L\,\rightarrow \, S^1$. As explained in \cite{Grabowska:2024}, it is possible to describe this line bundle as the quotient of the trivial line bundle $\R^2\,\rightarrow\,\R$ under the action of the group $\Z$ given by
\begin{equation*}
k.(x,z) = (x+k,(-1)^k z)\,.
\end{equation*}
The line bundle structure is inherited from the quotient $L=\R^2/ \Z$, which is a non-trivializable line bundle over the quotient $S^1 \simeq \R/ \Z $, and whose points are equivalence classes $\left[ (x,s) \right]$. The bundle structure can be defined by using two charts. For $x\notin \Z$, in every equivalence class, it can be chosen a representative with $x\in \left] 0, 1\right[$, so that the first chart is defined on $\cO = \left\lbrace [(x,z)]\in L \mid x\in \left] 0, 1\right[ \right\rbrace$ by
\[
\xi_1 \,\colon\, \cO \, \rightarrow \, \R^2\,, \qquad \xi_1(\left[ (x,z) \right]) = (x,z)\,.
\]
Analogously, if $x \neq k+\frac{1}{2}$, then it is possible to choose in every equivalence class a representative with $x\in \left] \frac{1}{2}, \frac{3}{2}\right[$, so that the second chart will be defined on $\cU = \Big\lbrace [(x,z)]\in L \mid x\in \left] \frac{1}{2}, \frac{3}{2}\right[ \Big\rbrace$ by
\[
\xi_2 \,\colon\, \cU \, \rightarrow \, \R^2\,, \qquad \xi_2\big(\left[ (x',z') \right]\big) = (x',z')\,.
\]
Clearly, $\cU \cup \cO \simeq B$, and the transition functions are defined as
\[
\xi_2 \circ \xi_1^{-1} (x,z) = (x,z) \,\quad\text{for}\quad x \in \left] 1/2\,,\, 1\right[ \,,
\]
and
\[
\xi_2 \circ \xi_1^{-1} (x,z) = (x+1, -z) \,\quad\text{for}\quad x \in \left] 0\,, \,1/2\right[ \,.
\]
Analogously, for the dual line bundle $L^*\,\rightarrow\,S^1$ we have local coordinates $(x,s)$ and $(x',s')$ and formally the same transition functions. Also, the vector bundle structure on $\sJ^1L$ can be described by using two charts. Indeed, using the canonical projection $\sj^1_0\,\zt \,\colon\,\sJ^1L \,\rightarrow \, L$, one obtains the new domains $\overline{\cO} = \big(\sj^1_0\,\zt\big)^{-1}(\cO)$ and $\overline{\cU} = \big(\sj^1_0\,\zt\big)^{-1}(\cU)$, with the coordinate functions
\[
\Xi_1 \,\colon\, \overline{\cO} \, \rightarrow \, \R^3\,, \qquad \Xi_1(\zl) = (x,p,z)\,,
\]
and
\[
\Xi_2 \,\colon\, \overline{\cU} \, \rightarrow \, \R^3\,, \qquad \Xi_2(\zl') = (x',p',z')\,,
\]
where $\zl \in \overline{\cO}$ and $\zl' \in \overline{\cU}$. The transition functions are easily derived,
\[
\Xi_2 \circ \Xi_1^{-1} (x',p',z') = (x,p,z) \,\quad\text{for}\quad x \in \left] 1/2\,,\, 1\right[ \,,
\]
and
\[
\Xi_2 \circ \Xi_1^{-1} (x',p',z') = (x+1,-p, -z) \,\quad\text{for}\quad x \in \left] 0\,, \,1/2\right[ \,.
\]
It is important to notice that, although the line bundle $L$ is not trivializable, the first jet bundle $\sJ^1 L$ is, because it is possible to choose two linearly independent non-vanishing global sections (see again \cite{Grabowska:2024}).
The corresponding transition functions for local coordinates $(x,\dot x,t)$ and $(x',\dot x',t')$ on $\sD L$ are
\[
(x',\dot x',t') = (x,\dot x,t) \,\quad\text{for}\quad x \in \left] 1/2\,,\, 1\right[ \,,
\]
and
\[
 (x',\dot x',t') = (x+1,\dot x, t) \,\quad\text{for}\quad x \in \left] 0\,, \,1/2\right[ \,,
\]
so the bundle $\sD L$ is trivial, $\sD L=\sT S^1\ti\R$.

\mn The Jacobi bracket on sections of the line bundle $L$ is described in terms of a homogeneous Poisson tensor on $L^*$, which is locally written as
\[
\zP_{\cO}(x,s) = E_{\cO}(x)\partial_x \we \partial_s\,,
\]
and
\[
\zP_{\cU}(x',s') = E_{\cU}(x')\partial_{x'} \we \partial_{s'}\,,
\]
so that we have a well-defined Poisson tensor if only $E_{\cO}(x)=E_{\cU}(x)$ for $x \in \left] \frac{1}{2}\,,\, 1\right[$, and $E_{\cU}(x+1)= -E_{\cO}(x)$ for $x \in \left] 0\,, \,\frac{1}{2}\right[$. In particular, we can choose the Poisson bivector, locally written as
\[
\zP_{\cO}(x,s) = \cos(\zp x)\partial_x \we \partial_s\,.
\]

\mn In the reduced picture, the fields of the Jacobi sigma model are pairs of maps $(\zF_0,\phi)$, with $\zF_0\,\colon\,\sT \zS \, \rightarrow \, \sJ^1 L$ being a vector bundle morphism, and $\phi\,\colon\, \zS \times \R \,\rightarrow \, L$ being a regular morphism of line bundles, both covering the same map $\phi_0\,\colon\, \zS \, \rightarrow \, S^1$. Using local coordinates $(u,v;r)$ on $\zS \times \R$ and local coordinates $(x,z)$ on $L$, we can express the morphism $\phi$ as
$$\phi(u,v,r)=\big(X(u,v),\mathfrak{z}(u,v)r\big),$$
where $\mathfrak{z}\ne 0$.

Analogously, we can express the morphism $\zF_0$ in terms of the triple $(X,\bp, \bz)$ where $X = x\circ \phi_0$ as before is a function on $\zS$ and $(\bp, \bz)$ are 1-forms on $\zS$ associated with the linear coordinates $(p,z)$ on $\sJ^1L$. It will be convenient to use also the map $\zf:\zS\to\Lst$ defined by $\phi$, in coordinates
$\zf(u,v)=\big(X(u,v),\s(u,v)\big)$, where $\s=1/\mathfrak{z}$.

\mn The Jacobi map $J^\#:\sJ^1L\to\sD L$ in coordinates $(x,p,z)$ and $(x,\dot x,t)$ reads
$$J^\#(x,p,z)=\big(x,\cos(\zp x)p,-\cos(\zp x)z\big),$$
and it has exactly the same form in coordinates $(x',p',z')$ and $(x',\dot x',t')$. The map $\sD_0\phi:\sT\zS\to\sD L$ is represented on the bases by $X$, and on the linear coordinates $\dot x$ and $t$ by the 1-forms $\xd X$ and $\xd\s/\s$ and our action functional for this Jacobi sigma model is
\[
S(X,\bp , \s, \z)=\int_\zS\Big(\s\p\we\xd X+\z\we\xd\s+ \s \cos(\zp X)\z\we\p\Big)\,.
\]
It is easy to see that the equations of motion are:
\bea
&&\xd X = \cos(\zp X)\bz\,; \nonumber\\
&&\xd \s = - \s\cos(\zp X)\bp\,; \nonumber\\
&&\xd \z = 0\,; \nonumber\\
&&\xd \p = -\zp\big(\sin(\zp X)\big)\z\we \p \nonumber\,.
\eea
They describe Lie algebroid morphisms $\zF_0:\sT\zS\to\sJ^1L$ such that $\sD_0(\phi)=J^{\#} \circ \zF_0$, the last requirement corresponds to the first two equations of motion.
\end{example}

\section{Almost Poisson and almost Jacobi sigma models}\label{AlmJac}
Natural generalizations of the concept of Lie algebroid have been introduced in \cite{Grabowski:1997a,Grabowski:1999} (see also \cite{Grabowski:2013a}), called in \cite{Grabowski:1999} \emph{general} and \emph{skew-symmetric algebroids}, \emph{skew algebroids} in short.
Like Lie algebroids, general (resp., skew) algebroid structures on a vector bundle $\zt:E\to M$  correspond to linear 2-contravariant tensor fields $\zP$ on the dual bundle $\zp:E^*\to M$, which are assumed to be bivector fields in the skew-symmetric case. The Jacobi identity is not assumed for such algebroids, so $\zP$ need not be a Poisson tensor. Of course, any 2-contravariant tensor field $\zL$ on a manifold $M$ induce an $\R$-bilinear bracket $\pb_\zL$ on functions on $M$ in the standard way: $\{f,g\}=\zL(\xd f,\xd g)$. This bracket does not satisfy the Jacobi identity, but the Leibniz rule
$$\{f,gh\}_\zL=g\{f,h\}_\zL+\{f,g\}_\zL h$$
holds. Of course, $\pb_\zL$ is skew-symmetric if and only if $\zL$ is skew-symmetric (i.e., it is a bivector field). In \cite{Grabowski:1999}, the brackets $\pb_\zL$ are called \emph{Leibniz brackets}, and \emph{skew-symmetric Leibniz brackets} in the case $\z$ is skew-symmetric. For the latter, the name \emph{almost Poisson brackets} is nowadays commonly used.

\mn General and skew algebroids have been applied in geometric mechanics to define the corresponding generalizations of Hamiltonian and Lagrangian formalisms \cite{Grabowska:2008,Grabowska:2006} in the Tulczyjew's picture. Skew algebroids also form a useful description of nonholonomic systems \cite{Grabowski:2009a}. In what follows, we will consider only skew algebroids.
\begin{definition} A \emph{skew algebroid} structure on a vector bundle $E$ is given by a linear bivector field
$\zP$ on $E^*$.
\end{definition}
\no Any linear bivector field has in local coordinates the form (\ref{SA}).
We have also a full analog of Theorem \ref{Lal} (cf. \cite{Grabowski:2013a,Grabowski:1999}) with the only difference that the Lie algebroid bracket need not satisfy the Jacobi identity and, consequently, the differential $\xd_\zP$ is not homological, so generally $\xd_\zP^2\ne 0$.
In complete analogy with the Lie algebroid case, skew algebroid morphisms are VB-morphisms inducing pullbacks that intertwine the skew algebroid de Rham differentials (see (\ref{mor})).

\mn Since the complete tangent lift is defined for all tensor fields \cite{Grabowski:1995,Yano:1973}, the lifted tensor $\zP=\dt\zL$ is a linear almost Poisson tensor on $\sT^*M$ for every bivector field $\zL$ on $M$, so $\zP$ defines a skew algebroid structure on $\sT^*M$. The functional (\ref{af}), acting on VB-morphisms $\zF:\sT\zS\to\sT^*M$, makes sense also for an almost Poisson $\zL$. As the Jacobi identity is not used to determine the field equations (\ref{fe}), we still have equations (\ref{fe}) for an almost Poisson $\zL$, which say that $\zF$ is a skew algebroid morphism. This way, we get for free an extension of the Poisson sigma model to an almost Poisson sigma model. Of course, as $\sT\zS$ is actually a Lie algebroid, non-trivial solutions (skew algebroid morphisms) exist only for particular almost Poisson tensors $\zL$. In particular, the differential forms
\[
\big(\zL^{ji}\circ\zf\big)\zh_j\quad\text{and}\quad
\big(\zL^{ij}_{,k}\circ\zf\big)\zh_i\we \zh_j
\]
must be closed (cf. (\ref{moreq})). Note that this is satisfied automatically if $\zL$ is Poisson.

\mn Consequently, in the framework of Jacobi bundles, we replace homogeneous Poisson structures with homogeneous bivector fields on $\Lst$. Equivalently, we consider \emph{almost Jacobi brackets} on sections of the line bundle $L\to M$, i.e., skew symmetric brackets $\pb$ on $\Sec(L)$ such that $\{\zs,\cdot\}:\Sec(L)\to\Sec(L)$ is a VB-derivation (linear first-order differential operator). We define further \emph{almost Jacobi algebroids}, their morphisms, and \emph{almost Jacobi sigma models} in an obvious way.
\begin{example} Consider on $\R^3$ with coordinates $(x,y,z)$ the almost Poisson tensor
$$\zL=\pa_x\we\big(\pa_y+x\,\pa_z\big).
$$
It is easy to see that it is not Poisson. Its tangent lift $\dt\zL$ in the adapted coordinates $(x,y,z,\dot x,\dot y,\dot z)$ in $\sT\R^3$ reads
$$\zP=\dt\zL=\dot x\,\pa_{\dot x}\we\pa_{\dot z}+\pa_{\dot x}\we\big(\pa_y+x\pa_z\big)-x\,\pa_{\dot z}\we\pa_x-\pa_{\dot y}\we\pa_x.
$$
Hence, the anchor $\zr=\zL^\#$ of the corresponding skew-algebroid structure on $\sT^*\R^3$ is
\beas
\zr(\xd x)&=&\pa_y+x\pa_z,\\
\zr(\xd y)&=&-\pa_x,\\
\zr(\xd z)&=&-x\,\pa_x,
\eeas
and the only nontrivial bracket among $\xd x,\xd y,\xd z$ is
$$[\xd z,\xd x]_\zP=\xd x.
$$
The \emph{characteristic distribution}, i.e., the image of the anchor map $\zr:\sT^*\R^3\to\sT\R^3$, is in this case a rank 2 distribution
\be\label{dist}\cD=\on{span}\{\pa_x,\pa_y+x\pa_z\},
\ee
on $\R^3$ which is clearly not involutive (actually a contact one).

\mn Consequently, the skew-algebroid de Rham derivative $\xdp$ on $\mathcal A(\sT\R^3)$ is completely characterized by
\bea\label{fed1} &\xdp(x)=-\pa_y-x\,\pa_z,\quad \xdp(y)=\pa_x,\quad \xdp(z)=x\,\pa_x,\\
&\xdp(\pa_x)=\pa_z\we\pa_x,\quad \xdp(\pa_y)=0,\quad\xdp(\pa_z)=0.\label{fed2}
\eea
In the cotangent bundle $\sT^*\R^3$, with the dual coordinates $(x,y,z,p_x,p_y,p_z)$, the linear functions $(p_x,p_y,p_z)$ represent the sections $(\pa_x,\pa_y,\pa_z)$ of $\sT\R^3$. If $\zF:\sT\zS\to\sT^*\R^3$
is a VB-morphisms covering a smooth map $\zf:\zS\to\R^3$, then we use the standard notation
$$X=x\circ\zf,\quad Y=y\circ\zf,\quad Z=z\circ\zf$$
for functions on $\zS$ being pull-backs of coordinates on $\R^3$, and
$$\zh_x=p_x\circ\zF,\quad \zh_y=p_y\circ\zF,\quad \zh_z=p_z\circ\zF$$
for the corresponding 1-forms on $\zS$.
The action functional in this case reads
$$S(\zF)=\int_\zS\Big(\zh_x\we\xd X+\zh_y\we\xd Y+\zh_z\we\xd Z+X\zh_x\we\zh_z+\zh_x\we\zh_y\Big)
$$
and the field equations are (cf. (\ref{fed1}) and (\ref{fed2}))
\beas &\xd X=-\zh_y-X\,\zh_z,\quad \xd Y=\zh_x,\quad \xd Z=X\,\zh_x,\\
&\xd\zh_x=\zh_z\we\zh_x,\quad \xd\zh_y=0,\quad\xd\zh_z=0.
\eeas
Consequently,
$$\xd\zh_x=0,\quad \xd X\we\zh_z=0,\quad  \xd X\we\zh_x=0. $$
Let us consider two cases: $\xd X\ne 0$ and $\xd X=0$.

\mn In the first case, the identity $\xd X\we\zh_x=0$ implies that $\zh_x$ is proportional to $\xd X$, and from $\xd Y=\zh_x$ it follows that $Y$ is a function of $X$ only, $Y=f(X)$, so $\zh_x=f'(X)\,\xd X$. We have then $\xd Z=X\cdot f'(X)\,\xd X$, so $Z=X\cdot f(X)-g(X)$, where $g'=f$.

From $\xd X\we\zh_z=0$ we get that also $\zh_z$ is proportional to $\xd X$ and, being closed, it is of the form $\zh_z=h(X)\,\xd X$ for some function $h$. Finally,
$$\zh_y=-\xd X-X\zh_z=-(1+h(X))\,\xd X.$$
Now, we can describe skew algebroid morphisms as follows.

\mn Take a function $X:\zS\to\R$ such that $\xd X\ne 0$, and functions $g,h:\R\to\R$.
Let $\zvy=\zi_{\xd X}$ be the liner function on $\sT\zS$ corresponding to $\xd X$.
Viewing the function $X$ as a basic function on $\sT\zS$, the VB-morphism $\zF:\sT\zS\to\sT^*\R^3$, given in coordinates $(x,y,z,p_x,p_y,p_z)$ on $\sT^*\R^3$ by
$$\zF=\Big(X,g'\circ X,X\cdot(g'\circ X)-g\circ X,(g''\circ X)\cdot\zvy,(h\circ X)\cdot\zvy,-(1+h\circ X)\cdot\zvy\Big),$$
is a skew algebroid morphism.

\mn Let us consider now the case $\xd X=0$, say $X=1$, and assume $\xd Y\ne 0$. It follows immediately that
$\xd Y=\xd Z=\zh_x$ and $\zh_z=-\zh_y$. Moreover, $0=\xd\zh_x=\zh_z\we\xd Y$, so, like above, $\zh_z=-f'(y)\xd Y$ and $\zh_y=f'(Y)\xd Y$ for some function $f:\R\to\R$. This time the construction of a skew algebroid morphism is the following. Take a regular function $Y:\zS\to\R$, $\xd Y\ne 0$, and a function $f:\R\to\R$. Let $\zvy=\zi_{\xd y}$ be the linear function on $\sT\zS$ corresponding to $\xd Y$. Viewing the function $Y$ as a basic function on $\sT\zS$, the VB-morphism $\zF:\sT\zS\to\sT^*\R^3$, given in coordinates $(x,y,z,p_x,p_y,p_z)$ on $\sT^*\R^3$ by
$$\zF=\Big(1,Y,Y+c,\zvy,(f'\circ Y)\cdot\zvy,-(f'\circ Y)\cdot\zvy\Big),$$
where $c\in\R$ is an arbitrary constant, is a skew algebroid morphism.

\mn Note that the image of $\zS$ in $\R^3$ under $\zf$ is 1-dimensional in our constructions. This must be so, as $\zF$ must intertwine the anchor maps, so this image must be an integral submanifold of the characteristic distribution. But the characteristic distribution in our case (cf. (\ref{dist})) is of rank 2 and `maximally nonintegrable', so the integrable submanifolds cannot be of dimension two.
\end{example}

\section{Conclusions and outlook}
In this paper we have proposed three variants of Jacobi sigma models which turned out to be essentially equivalent. Starting from the observation that a Jacobi structure can be viewed as a homogeneous Poisson structure on a principal $\Rt$-bundle, an action functional for a Jacobi sigma model has been defined as a homogeneous Poisson sigma model. An important particular case is that of \emph{contact sigma models}. The fields are morphisms of vector bundles from $\sT\zS$ to $\sT^*\Lst$, where $\Lst$ is a principal $\Rt$-bundle endowed with a homogeneous Poisson bivector field of degree -1. Indeed, since a general Jacobi bracket is a bracket on the space of sections of a line bundle, the description of this bracket in terms of a bivector field and a Reeb vector field is only a local one. The use of the unfolded space $\sT^*\Lst$, however, allowed us to define an action functional with the necessary invariance properties. The solutions of the model are morphisms of vector bundles which are also Lie algebroid morphisms. We also presented the standard Hamiltonian formulation of the system for $\zS$ being a rectangle. In this case, the space of solutions of the model is finite-dimensional and it possesses the structure of a symplectic groupoid. However, this groupoid inherits the action of the group $\Rt$ from the corresponding Lie algebroid, and the symplectic structure is homogeneous of degree 1 with respect to this action. Adopting the point of view elaborated in \cite{Bruce:2017}, this groupoid is nothing but a contact groupoid.

Our homogeneous approach generalizes the approach `a la Kaluza-Klein' proposed in \cite{Chatzistavrakidis:2020} in the sense that it extends Jacobi brackets defined on functions to brackets defined on sections of possibly nontrivial line bundles. Pushing forward this idea, one could replace the action of the group $\Rt$ with other structure groups: this approach could be useful to construct new models coupling gauge symmetries associated with internal degrees of freedom and diffeomorphisms of the surface $\zS$. Further work in this direction would shed some light on the relations of such an approach with those proposed in \cite{Bojowald:2005}, where groupoid actions appear in order to formulate theories with the above mentioned symmetries, or in \cite{Ikeda:2019}, where a general analysis of first-class constraints is presented.
Since the homogeneous structure plays a crucial r\^ole when dealing with a nontrivial Jacobi bundle, an alternative approach to Jacobi sigma models is possible, this time for a model with morphisms of principal $\Rt$-bundles as fields whose solutions are morphisms of Jacobi algebroids. Because the Poisson structure on $\Lst$ is homogeneous, the tangent bundle $\sT^*\Lst$ is not only a Lie algebroid like in the standard case, but it carries a canonical $\Rt$-action itself, and $\Rt$ acts by Lie algebroid automorphisms; this is the structure of a Jacobi algebroid. However, as we have proved, there is a one-to-one correspondence between the solutions of this model and the solutions of the homogeneous approach.

Finally, in our third approach, the fields take values in the reduced data, i.e., they are represented by VB-morphisms of $\sT\zS$ to the first jet bundle $\sJ^1L$ which also carries a Lie algebroid structure inherited from $\sT^*\Lst$. This picture does not refer to the `Poissonization' of Jacobi structures, but again, it is equivalent to the first two models.

\mn Coming back to the solutions of the homogeneous approach, they are Lie algebroid morphisms. Therefore, when we look from a dual point of view, they are coisotropic submanifolds of the Poisson manifold $\sT \Lst \times \sT^*\zS$, with respect to the Poisson bracket which is the product of the Poisson structures defined on the duals of the Lie algebroids $\sT^* \Lst$ and $\sT\zS$, respectively (see \cite{Mackenzie:2005}). This approach shares some similarities with the one used in \cite{Witten:1988}, where fields of nonlinear sigma models are thought as sections of a suitable $M$-bundle, where $M$ is the target manifold of the model. As a future line of research, one could investigate the formulation of a dual approach to a nonlinear sigma model using dual relations to vector bundle morphisms. In this case, however, fields need not be sections, even if the support bundle would be trivial. Once identified a solution as a coisotropic submanifold of the whole Poisson manifold, then a non-commutative deformation of this solution could be constructed in terms of $A_{\infty}$ structures integrating a proper $P_{\infty}$ structure, as presented in \cite{Cattaneo:2007}.

Eventually, we have also introduced a generalization involving morphisms of skew-algebroids $(E,\zP)$. In this case, we consider brackets that do not satisfy the Jacobi identity, whose `Poissonization' leads to almost Poisson brackets. From the dual point of view, the corresponding graded `de Rham derivative' of the algebra $\A(E^*)$ is not a homological derivation. In the literature, almost Jacobi structures appear when considering a twist of the action functional via a WZW term involving a non-closed $H$-field \cite{Chatzistavrakidis:2020, Bascone:2024}. The action functional appearing in Section \ref{AlmJac}, however, is different from the models derived in this way. Further investigation is required to understand the relation with these models and also with the generalized approach to sigma models proposed in \cite{Chatzistavrakidis:2023}.

\section*{Acknowledgements}
FDC thanks the UC3M, the European Commission through the Marie Sk\l odowska-Curie COFUND Action (H2020-MSCA-COFUND-2017- GA 801538), and Banco Santander for their financial support through the CONEX-Plus Programme. He also thanks the Institute of Mathematics of the Polish Academy of Sciences for its kind hospitality during the development of this project.

The research of KG and JG was partially funded by the National Science Centre (Poland) within the project WEAVE-UNISONO, No. 2023/05/Y/ST1/00043.

\vskip.5cm
\noindent Fabio Di Cosmo\\\emph{Departamento de F\'{i}sica y Matem\'atica, Universidad de Alcal\'a}\\
{\small Ctra Madrid-Barcelona, 33,600. 28805 Alcal\'a de Henares, Madrid, Spain}\\
{\tt fabio.di@uah.es}\\ https://orcid.org/0000-0003-0256-5913
\\

\no Katarzyna Grabowska\\\emph{Faculty of Physics,
University of Warsaw,}\\
{\small ul. Pasteura 5, 02-093 Warszawa, Poland} \\{\tt konieczn@fuw.edu.pl}\\
https://orcid.org/0000-0003-2805-1849\\

\noindent Janusz Grabowski\\\emph{Institute of Mathematics, Polish Academy of Sciences}\\{\small ul. \'Sniadeckich 8, 00-656 Warszawa, Poland}\\
{\tt jagrab@impan.pl}\\  https://orcid.org/0000-0001-8715-2370
\end{document}